\documentclass[10pt,journal,compsoc,onecolumn]{IEEEtran}
\ifCLASSOPTIONcompsoc
  \usepackage[nocompress]{cite}
\else
  \usepackage{cite}
\fi
\usepackage{subfigure,cite,graphicx,amsmath,algorithmic,amssymb,mathrsfs,epsf}
\usepackage[usenames]{color}
\usepackage{multirow}
\usepackage{tabularx}
\usepackage{fixltx2e}
\usepackage{epstopdf}
\usepackage{textcomp}
\usepackage{tikz}
\usepackage{ctable}
\usepackage{lipsum}
\usepackage{enumitem}
\usepackage{amsthm}
\usepackage{verbatim}
\usepackage{xcolor}
\usetikzlibrary{calc}
\usepackage{soul}
\usepackage{pgfplots}
\usepackage{stfloats}
\usepackage{setspace}
\usepackage{caption}
\usepackage{algorithm}
\newlength\myindent
\begin{document}
\sloppy
\title{A Personalized Preference Learning Framework for Caching in Mobile Networks}
\author{Adeel~Malik, Joongheon Kim,~\IEEEmembership{Senior Member,~IEEE}, Kwang~Soon~Kim,~\IEEEmembership{Senior Member,~IEEE}, \\ and
        Won-Yong Shin,~\IEEEmembership{Senior Member,~IEEE}
\thanks{A. Malik is with the Department of Communication System, EURECOM, Sophia-Antipolis 06904, France (e-mail: adeel\_malik91@yahoo.com).}
\thanks{J. Kim is with the School of Electrical Engineering, Korea University, Seoul 02841, Republic of Korea (e-mail: joongheon@korea.ac.kr).}
\thanks{K. S. Kim is with the Department of Electrical and Electronic Engineering, Yonsei University, Seoul 03722, Republic of Korea (e-mail: ks.kim@yonsei.ac.kr.).}
\thanks{W.-Y. Shin (corresponding author) is with the Department of Computational Science and Engineering, Yonsei University, Seoul 03722, Republic of Korea (e-mail: wy.shin@yonsei.ac.kr).}}

\newtheorem{axiom}{Axiom}
\newtheorem{lemma}{Lemma}
\newtheorem{theorem}{Theorem}
\newtheorem{prop}{Proposition}
\newtheorem{observation}{Observation}
\newtheorem{definition}{Definition}
\newtheorem{remark}{Remark}
\newtheoremstyle{case}{}{}{}{}{}{:}{ }{}
\theoremstyle{case}
\newtheorem{case}{Case}

\def \diag{\operatornamewithlimits{diag}}
\def \min{\operatornamewithlimits{min}}
\def \max{\operatornamewithlimits{max}}
\def \log{\operatorname{log}}
\def \max{\operatorname{max}}
\def \rank{\operatorname{rank}}
\def \out{\operatorname{out}}
\def \exp{\operatorname{exp}}
\def \arg{\operatorname{arg}}
\def \E{\operatorname{E}}
\def \tr{\operatorname{tr}}
\def \SNR{\operatorname{SNR}}
\def \dB{\operatorname{dB}}
\def \ln{\operatorname{ln}}

\def \bmat{ \begin{bmatrix} }
\def \emat{ \end{bmatrix} }

\def \be {\begin{eqnarray}}
\def \ee {\end{eqnarray}}
\def \ben {\begin{eqnarray*}}
\def \een {\end{eqnarray*}}

\IEEEtitleabstractindextext{
\begin{abstract}
This paper comprehensively studies a content-centric {\em mobile} network based on a {\em preference learning} framework, where each mobile user is equipped with a finite-size cache. We consider a practical scenario where each user requests a content file according to its own preferences, which is motivated by the existence of heterogeneity in file preferences among different users. Under our model, we consider a single-hop-based device-to-device (D2D) content delivery protocol and characterize the average hit ratio for the following two file preference cases: the {\em personalized} file preferences and the {\em common} file preferences. By assuming that the model parameters such as user activity levels, user file preferences, and file popularity are unknown and thus need to be inferred, we present a {\em collaborative filtering (CF)}-based approach to learn these parameters. Then, we reformulate the hit ratio maximization problems into a submodular function maximization and propose two computationally efficient algorithms including a greedy approach to efficiently solve the cache allocation problems. We analyze the computational complexity of each algorithm. Moreover, we analyze the corresponding level of the approximation that our greedy algorithm can achieve compared to the optimal solution. Using a real-world dataset, we demonstrate that the proposed framework employing the personalized file preferences brings substantial gains over its counterpart for various system parameters.
\end{abstract}

}
\maketitle
\IEEEdisplaynotcompsoctitleabstractindextext

%
\IEEEpeerreviewmaketitle
\section{Introduction}~\label{section:1}
\IEEEPARstart{T}{h}e growing trend in mobile data traffic drives a need of a new wireless communication technology paradigm in which radios are capable of learning and decision making in order to autonomously determine the optimal system configurations. In this context, equipping the communication functionality with machine learning-based or data-driven algorithms has received a considerable attention both in academia as well as in industrial communities.

\subsection{Prior Work}~\label{section:11}

Cache-enabled (or content-centric) wireless systems are equipped with finite storage capacity, which restricts us from storing the entire content library at the local cache. In order to bring content objects closer to requesting users, deciding which content needs to be cached at which user's or helper's cache plays a crucial role on the overall performance of content-centric wirelss networks. This cache placement problem has recently attracted a wide attention. In general, the prior work on caching can be divided into two categories: caching at small-cell base stations (or helper nodes) \cite{femto, G1, G2, G3} and caching at users (or devices) \cite{d2d1,hit,d2d, malik}. Since the optimal caching problems in \cite{femto, G1, G2, G3, d2d1, hit} cannot be solvable in polynomial time, approximate solutions with performance guarantees were presented. It was shown in \cite{femto, d2d1} that the optimal cache placement problem fell in the category of monotone submodular maximization over a matroid constraint~\cite{greedy1}, and a proposed time-efficient greedy algorithm achieved an approximation within the factor of $\frac{1}{2}$ of the optimum. In \cite{d2d1}, a dynamic programming algorithm whose complexity increases exponentially with the number of users was also proposed to obtain the optimal solution. Similarly, a joint routing and caching problem~\cite{G1}, a multicast-aware caching problem~\cite{G2}, and a mobility-aware caching problem~\cite{G3} were studied, while proving the NP-hardness of their caching problems and presenting polynomial-time greedy algorithms to obtain their approximate solutions.

On the one hand, studies on integrating {\em machine learning-based or data-driven} algorithms into caching frameworks in wireless networks have recently received attention \cite{ learn3, learn4, learn6, learn7 }.  In \cite{learn3, learn4}, the optimal caching problems were studied when the file popularity profile is unknown---each problem was modeled as a combinatorial multi-armed bandit problem to learn the content popularity distribution for caching at one base station. Moreover, a transfer learning-based technique exploiting the rich contextual information (e.g., users' content viewing history, social ties, etc.) was proposed in \cite{learn6} to estimate the content popularity distribution. The training time in a transfer learning-based approach to obtaining a good estimate of the content popularity distribution for caching was also investigated in \cite{learn7}.

On the other hand, the importance of personalized file preferences in content-centric networks  was studied in~\cite{hit, chen, learn1, learn2}. In \cite{hit}, a low-complexity semigradient-based cooperative caching scheme was designed in mobile social networks by incorporating probabilistic modeling of user mobility and heterogeneous interest patterns.  It was assumed in \cite{hit} that all users have the same activity level (i.e., the same number of requests generated by each user) and the user file preferences are known, which may be hardly realistic. Later, the optimal caching policy via a greedy algorithm for cache-enabled device-to-device (D2D) communications was presented in \cite{chen} by modeling the behavior of user requests resorting to probabilistic latent semantic analysis. In~\cite{learn1}, when there is some coordination between recommender systems and caching decisions in a wireless network with small cells, possible benefits that can arise for the end users were investigated. In~\cite{learn2}, a recommender system-aided caching placement scheme was also developed in wireless social networks. Recently, the benefits of user preference learning over the common file preferences were studied in wireless {\em edge-caching} scenarios, where the network edge such as fog access points~\cite{pf1}, small-cell base stations~\cite{pf2}, and macro base stations~\cite{pf3} is taken into account so that users fetch the content of interest from edge caches. However, different from existing literature in~\cite{pf1, pf2, pf3}, our focus is on {\em D2D caching at mobile users} by inferring different user activity levels and user file preferences for each user.  Moreover, in~\cite{Inc1,Inc2,Inc3,Inc4}, there have been various attempts to tackle caching problems by taking into account the concern that several issues such as limited memory~\cite{malik}, security~\cite{Inc5}, limited energy \cite{Inc6}, and social ties~\cite{Inc5} may discourage mobile users to share their content files with other users over the D2D transmission. In this context, proper incentive mechanisms (e.g., trust based on friendships~\cite{Inc2} and monetary incentives~\cite{Inc3}) for those users who are willing to share their files via D2D communications were introduced offered by operators.

As another caching framework, coded caching \cite{Maddah1, Maddah2, d2d} and maximum distance separable (MDS)-coded caching \cite{ femto, malik} have received a lot of attention in content-centric wireless networks. The content placement is optimized so that several different demands can be supported simultaneously with a single coded multicast transmission, while MDS-coded subpackets of content objects are stored in local caches and the requested content objects are retrieved using unicast transmission.

\subsection{Main Contribution}~\label{section:12}

In this paper, we present a {\em preference learning}
framework for a content-centric {\em mobile} caching network in which mobile users equipped with a finite-size cache are served by one central processor (CP) having access to the whole file library. Each user requests a content file from the library independently at random according to its own {\em predicted personalized file preferences}. Unlike most of the prior work (e.g.,~\cite{alfano, anh, milad, malik} and references therein) in which a user requests according to the common file popularity distribution, our model is motivated by the existence of heterogeneity in file preferences among different mobile users, e.g., applications involving on-demand videos requested by mobile users. In our mobile network model, we consider single-hop D2D content delivery and characterize the average hit ratio. Then, we aim to develop caching strategies in the sense of maximizing the hit ratio (also known as the offloading ratio) for two file preference cases: the {\em personalized} file preferences and the {\em common} file preferences.  The most distinct feature in our framework is to incorporate the notion of {\em collaborative filtering (CF)}, which was originally introduced in the design of recommender systems~\cite{ ubcf,ibcf,svd,auto}, into caching problems in order to infer/predict the unknown model parameters such as user activity levels, user file preferences, and file popularity. Despite the existence of several techniques to predict the file popularity, CF is the most widely used method to predict users' interests in a particular content object (e.g., the content ratings in recommender systems), since it can be implemented easily without requiring in-depth domain knowledge, producing satisfactory prediction accuracy.  The main idea of CF is to exploit users' past behaviors and infer the hidden user preferences. We aim at showing a successful application of such a well-known technique in one domain (e.g., recommender systems) to another domain (e.g., content-centric wireless systems) by predicting our model parameters via CF. We note that the accessing frequency of a file in \cite{chen} may not be a suitable metric for predicting unknown model parameters. This is because content files such as season episodes or news talk shows do not need to be accessed more than once as users tend to watch them only once. In our study, we thus exploit the user rating history as a metric for learning the above model parameters instead of using the number of times that each user accesses a file. The main technical contributions of this paper are summarized as follows:
\begin{itemize}
\item We first introduce a file preference learning framework in cache-enabled mobile network, which intelligently combines caching and CF techniques.
\item We characterize the average hit ratio under our content-centric mobile network according to both personalized file preferences and common file preferences.
\item We present a CF-based approach to infer the model parameters such as user activity levels, user file preferences, and file popularity by utilizing the user rating history.
\item We formulate hit ratio maximization problems for our two caching strategies and analyze the computational complexity of the optimal solutions.
\item We reformulate the hit ratio maximization problems into a submodular function maximization subject to matroid constraints in order to obtain an approximate solution within provable gaps.
\item We present two computational efficient algorithms, including a {\em greedy}  approach, and analyze corresponding computational complexities. Additionally, we show the corresponding level of the approximation that our greedy algorithm can achieve compared to the optimal solution.
\item Using a real-world dataset adopted for evaluating recommender systems, we perform intensive numerical evaluation and demonstrate that 1) the proposed framework employing the personalized file preferences leads to substantial gains over the baseline employing the common file preferences with respect to various system parameters and 2) the performance of our computationally efficient approaches is quite comparable to the optimal one while they manifest significant complexity reduction.
\end{itemize}
 
 To the extent of our knowledge, this paper is the first attempt to present {\em personalized parameter learning} in cache-enabled {\em mobile} networks.
 
\begin{table}[t]
\centering
\caption{ Summary of Notations}
\label{tab:Not}
\begin{tabular}{ll}
\hline
Parameters & Description \\ \hline
$N$         &      Number of users      \\
  $S$    &      Cache size of each user   \\
$F$      &      Library size       \\
 $\mathcal{D}$     &    Collaboration distance         \\
$\mathcal{T}$      &    Time period for cache update         \\
 $P(u_i)$     &          User activity level   \\
$P(f_j|u_i)$      &     User file preference        \\
$P(f_j)$      &       File popularity      \\
 $\mathcal{C}^{N\times F}$         &   Cache allocation strategy         \\
   $H_{avg}$               &   Average hit ratio         \\
     $p_{u_i,u_j} $                   &    Contact probability between users $u_i$ and $u_j$       \\
        $O_{avg}$                        &          Average outage probability    \\
       $\Lambda$                    &        Outage capacity     \\
     ${\bf R}_c^{N\times F}$                 &    Rating  matrix          \\
       $r_{k,a}$                       &      Rating given by user $u_k$ to file $f_a$        \\
       $n_k$                   &      Number  of  ratings  given  by  user $u_k$        \\ \hline
 \end{tabular}
\end{table}

\subsection{Organization}~\label{section:13}
The rest of this paper is organized as follows. In Section 2, the network and caching models are described. In Section 3, problem definitions for both proposed and baseline caching strategies are presented. In Section 4, a CF-based parameter learning method is introduced. The optimal and greedy caching algorithms are presented with complexity analysis in Section 5. In Section 6, numerical results are shown to analyze the impact of personalized file preferences in our network. Finally, Section 7 summarizes the paper with some concluding remarks.

\subsection{Notations}~\label{section:14}
Throughout this paper, $P(\cdot)$ is the probability, $\emptyset$ is the empty set, and $\left|\mathcal{X}\right|$ denotes the cardinality of a set $\mathcal{X}$. We use the asymptotic notation $f(x)= \mathcal{O}(g(x))$ means that there exist constants $a$ and $c$ such that $f(x) \leq ag(x)$ for all $x > c$. Moreover, Table~\ref{tab:Not} summarizes the notations used throughout this paper. Some notations will be more precisely defined in the following sections, as we introduce our system model and proposed approaches.

\section{System Overview}~\label{section:2}
\setlength{\abovedisplayskip}{3pt}   
\setlength{\belowdisplayskip}{3pt}    

\subsection{Network Model}\label{section:21}
Let us consider a wireless network illustrated in Fig.~\ref{fig:SM}, having a CP serving $N$ mobile users denoted by the set $\mathcal{N}= \left\{u_1, u_2, \cdots, u_N\right\}$. In our content-centric mobile network, we assume a practical scenario where each mobile user $u_i \in \mathcal{N}$ is equipped with a local cache of the same finite storage capacity $S$ (i.e., the cache can store up to $S$ content files) and acts as a helper to share the cached files via D2D {\em single-hop} communications. We assume that all the users are willing to share the content files with other users altruistically over the D2D transmission all the time\footnote{We leave how to incentivize the users in another setting with user selfishness as future work since the design of incentive mechanisms for file sharing in content-centric networks is not straightforward and needs to be judiciously incorporated into the overall caching framework including D2D content delivery.}. The CP is able to have access to the content library server (e.g., YouTube) via infinite-speed backhaul and initiates the caching process. We consider a practical scenario that the number of content files in the library is dynamic and keeps updated due to the continuous addition of new user-generated content (UGC). It has been observed that the traffic toward the content library server is dominated by the small portion of popular content files~\cite{traffic}.
To adopt this behavior, we assume that users are interested in requesting the content file from a set of $F$ popular content files $\mathcal{F} =\left\{f_1, f_2, \cdots, f_F\right\}$ that dominates the request traffic, where the size of each file is the same.
\begin{figure}[t]
\centering
 \includegraphics[width=0.5\linewidth]{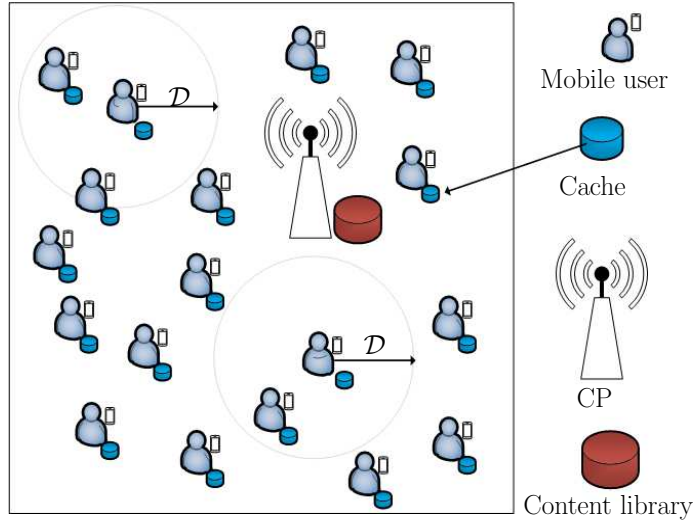}
\caption{Our content-centric mobile network model.}
\label{fig:SM}
\end{figure}

In our network model, the time is divided into independent slots $t_1, t_2,\cdots$ and each user generates requests for content files during its allocated time slot. Note that each user can generate the request of content files that are not in the set $\mathcal{F}$, but for our content-centric network, we are interested in caching only the popular content files that dominate the request traffic.

In addition, we would like to illuminate two important real-world observations that were largely overlooked in most studies on  cache-enabled wireless networks~\cite{alfano,anh,milad,malik}: the file popularity is not the same as the preference of each individual user; and only a small portion of users are active in creating the data traffic. On the other hand, the main focus of this study is on comprehensively studying the impact and benefits of {\em personalized} file preferences in our {\em mobile} network by assuming that each user $u_i\in \mathcal{N}$ has a different activity level $P(u_i)$ and a different personalized file preference $P(f_j|u_i)$. To this end, we formally define three model parameters as follows:

\begin{itemize}
    \item \textbf{Activity level $P(u_i)$} of user $u_i$ is defined as the probability that user $u_i\in \mathcal{N}$ is active by requesting a content file, where $P(u_i) \in [0,1]$. More specifically, $P(u_i)$ indicates the ratio of the number of active time slots in which user $u_i$ sends a content request to the total number of time slots,  $\mathcal{T}$, in a specific time period.
    \item \textbf{Personalized file preference $P(f_j|u_i)$}  of user $u_i$ is defined as the conditional probability that user $u_i \in \mathcal{N}$ requests file $f_j \in \mathcal{F}$ given that the user sends a request, where $\sum_{f_j \in \mathcal{F}}P(f_j|u_i) = 1$ and $P(f_j|u_i) \in [0,1]$. More specifically, $P(f_j|u_i)$ indicates the ratio of the preference that user $u_i$ has given to file $f_j$ to the sum of all file preferences of the same user.
    \item \textbf{File popularity $P(f_j)$} is defined as the probability that a file $f_j$ is requested, where $\sum_{f_j \in \mathcal{F}} P(f_j) = 1$ and $P(f_j) \in [0, 1]$. This probability represents how popular a content file is in a network. Note that this is the case where all users have the \textit{common} file preferences (i.e., $P(f_j|u_a)=P(f_j|u_b)= P(f_j)$) as in~\cite{alfano,anh,milad,malik}.
\end{itemize}
It will be described in Section~\ref{section:4} how to learn these three parameters based on the CF-based inference approach. Moreover, due to the continuous addition of new UGC to the content library, the popularity of content files changes over time. In our study, we assume that the model parameters such as user activity levels and user file preferences are fixed for a certain period of $\mathcal{T}$ time slots, and after every period $\mathcal{T}$, the CP iteratively learns the model parameters based on the updated popularity so that it initiates the content placement phase to fill the cache of each user.

\subsection{Caching Model}\label{section:22}
In content-centric wireless networks, a caching mechanism can be divided into the following two stages: the content placement phase and the content delivery phase. We first describe the content placement phase which determines the strategy for caching the content objects in the cache of $N$ mobile users.
During the content placement process, each user $u_i \in \mathcal{N}$ stores the content files according to a caching vector $\textbf{C}_{u_i}= [c_{u_i,f_1}, c_{u_i,f_2}, \cdots, c_{u_i,f_F}]$ (will be optimized later on the basis of the personalized file preferences), where $c_{u_i,f_j}$ is given by
	\begin{equation} \label{eq:c_u}
c_{u_i,f_j} = \begin{cases}
    1 & $if the user $u_i$ caches file $f_j \\
    0 & $if the user $u_i$ does not cache cache $f_j.
\end{cases}
\end{equation}
In order to have a feasible cache allocation strategy $\mathcal{C}^{N\times F}=$ $ [\textbf{C}_{u_1}; \textbf{C}_{u_2}; \cdots; \textbf{C}_{u_N}]$, $\mathcal{C}^{N\times F}$ should satisfy the following constraints:
\begin{align}
	\sum_{f_j \in \mathcal{F}}&c_{u_i,f_j} \leq S, \hspace{1cm} \forall u_i \in \mathcal{N}, \label{eq:Cons1}\\
c_{u_i,f_j} &\in \left\{0,1\right\}, \hspace{1cm} \forall u_i \in \mathcal{N}, f_j \in \mathcal{F}. \label{eq:Cons2}
\end{align}

\noindent Note that the constraint in \eqref{eq:Cons1} is the individual user's storage capacity constraint. Due to the dynamics of the file popularity, the caching strategy $\mathcal{C}^{N\times F}$ should be iteratively determined by the CP based on the predicted model parameters after every time period $\mathcal{T}$ during the {\em off-peak time}, representing the time period when network resources are in abundance (i.e., the data traffic is very low). One example includes such a midnight time that the content request activity from the network is very low as most of mobile users are not active during this time. 

We now move on to the delivery phase which allows the requested content objects to be delivered from the source to the requesting user over wireless channels. In our network model, the users are assumed to prefetch a part of (popular) content objects in their local caches from the content library server or the CP when they are indoors. For example, during off-peak times, the CP can initiate the content placement phase and fill the cache of each user. On the other hand, for the case when the actual requests take place, we confine our attention to an outdoor environment where users are moving. This comes from the fact that in the fast mobility scenario, reception of large-scale files (e.g., a high-resolution video) through cellular networks, operating on licensed spectrum bands, may not be cost-effective. Thus, in the outdoor setting, a reasonable assumption is that only D2D communications are taken into account for content delivery, i.e., the server and base stations do not participate in the delivery phase, which essentially lies in the same line as in \cite{malik, alfano, d2d}. More specifically, in our content-centric mobile network, we focus on the {\em peak time} in which each mobile user retrieves the requested content file either from his/her own cache or via single-hop D2D communications from one of the users storing the requested content file in their caches within the collaboration distance $\mathcal{D}$ of the requested user (refer to Fig.~\ref{fig:SM}). In Section~\ref{section:hb}, we shall relax this assumption and numerically show the effects of deploying cache-enabled femto-cell base stations (FBSs) on the performance by considering the content retrieval on both D2D and cellular modes. In our study, the protocol model in~\cite{gupta} is adopted for successful D2D content transmission. According to the protocol model, the content delivery from source user $u_s$ to requesting user $u_d$ will be successful if the following conditions hold: 1) $d_{u_su_d}(t_i)\!\leq \mathcal{D}$ and 2) $d_{u_bu_d}(t_i)\geq(1+\Delta)\mathcal{D}$, where $d_{u_su_d}(t_i)$ and $d_{u_b,u_d}(t_i)$ represent the Euclidean distances between user pairs $(u_s,u_d)$ and $(u_b,u_d)$, respectively, at given time slot $t_i$, for every user $u_b$ that is simultaneously transmitting at the same time slot, and $\Delta>0$ is a guard factor. We assume that the number of users, $N$, in the network is sufficiently large, so that each square cell of area $\mathcal{D}^2$ contains at least one requesting user with high probability (refer to~\cite{gupta} for more details). When successful transmission occurs (i.e., the two conditions in the protocol model  hold), we assume that the total amount of data transferred during one time slot is large enough to transfer a complete content file from a sender to a receiver (requester). Nevertheless, in a given time slot, a requesting user can receive no more than one content file.

\section{Problem Definition}~\label{section:3}
In this section, we first define the average hit ratio as a performance metric used throughout this paper. Then, we introduce our problem formulation in terms of maximizing the average hit ratio in our content-centric mobile network employing the personalized file preferences. For comparison, we present a baseline strategy that employs the common file preferences. Finally, we show a way to reformulate the hit ratio maximization problems into  submodular function maximization subject to matroid constraints.

\subsection{Performance Metric}~\label{section:31}
The goal of our cache-enabled mobile network is to reduce the traffic caused by the infrastructure mode in which users retrieve their requested content through the CP, which will be enabled either during off-peak times or when the content is not available via D2D communications. To this end, the performance metric of interest is the average hit ratio $H_{avg}$ at any time slot and is defined as follows~\cite{chen, hit}.

\begin{definition}[Average hit ratio $H_{avg}$]~\label{df:havg}  Let $H_{u_i,f_a}$ denote the hit probability that a user $u_i \in \mathcal{N}$ requesting a content file $f_a \in \mathcal{F}$ at any time slot retrieves the requested content either from his/her own cache or from the cache of a user which falls in a given time slot within collaboration distance $\mathcal{D}$ between the requesting user $u_i$ and another user holding the requested content file. In this case, the average hit ratio $H_{avg}$ that the content file is successfully retrieved via D2D communications is defined as
 \begin{equation} \label{eq:c}
H_{avg} = \frac{1}{m}\sum_{u_i \in \mathcal{N}} \sum_{f_a \in \mathcal{F}} P(u_i) P(f_a|u_i) H_{u_i,f_a},
\end{equation}
where $m= \sum_{u_i \in \mathcal{N}} P(u_i)$ is the normalization constant, which is used so that $H_{avg}$ lies between 0 and 1.
\end{definition}

We aim at solving our caching problems in terms of maximizing the hit ratio in \eqref{eq:c}. However, we also demonstrate that our approach offers potential gains over the baseline method in terms of outage capacity defined as follows.

\begin{definition}[Outage Capacity $\Lambda$]~\label{df:lambda} Let $O_{u_i,f_a}$ denote the probability that there is no user within the collaboration distance $\mathcal{D}$ of user $u_i$ that is storing the content $f_a$ that the user has requested. Then, the average outage probability $O_{avg}$ is given by
 \begin{equation} \label{eq:oavg}
O_{avg} = \frac{1}{m}\sum_{u_i \in \mathcal{N}} \sum_{f_a \in \mathcal{F}} P(u_i) P(f_a|u_i)  O_{u_i,f_a}, 
\end{equation}
where $m= \sum_{u_i \in \mathcal{N}} P(u_i)$ is the normalization constant. Let $\lambda$ denote the maximum number of single-hop D2D communication links that can be established simultaneously in each time slot. Then, the outage capacity $\Lambda$ is defined as
 \begin{equation} \label{eq:oc}
\Lambda = \lambda \left(1- O_{avg}\right).
\end{equation}
\end{definition} 

Note that the outage capacity $\Lambda$ quantifies the maximum number of simultaneous requests that can be successfully served by local caches at mobile users over the entire network.

\subsection{Problem Formulation}~\label{section:32}
In this subsection, we present our problem formulation for the optimal cache allocation strategy that maximizes the average hit ratio $H_{avg}$ in~\eqref{eq:c} for the content-centric mobile network employing the personalized file preferences in Section~\ref{section:21}.

Let $p_{u_i,u_j}$ denote the contact probability that that users $u_i,u_j\in\mathcal{N}$ are within the communication distance $\mathcal{D}$ at a time slot. Then, given three probability distributions of the activity level $P(u_i)$ of each user, the personalized file preference $P(f_a|u_i)$ of each user, and the contact probability $p_{u_i,u_j}$, the average hit ratio $H_{avg}$ depends solely on the cache allocation strategy $\mathcal{C}^{N\times F}$. Among all the cache allocation strategies, the optimal one will be the one that leads to the maximum $H_{avg}$. Now, from the caching constraints in \eqref{eq:Cons1} and \eqref{eq:Cons2}, the optimal cache allocation strategy $\mathcal{C}^{N\times F}$ for the content-centric mobile network employing the personalized file preferences can thus be the solution to the following optimization problem: 

\begin{subequations} \label{eq:opp}
\begin{align} \label{eq:ofp}
\!\!\!  \textbf{[P1]} \underset{\mathcal{C}^{N\times F}}{\max} \!\!\! \sum_{u_i \in \mathcal{N}}\!\! \sum_{f_a \in \mathcal{F}}\!\!\! P(u_i) P(f_a|u_i)\!\!\left(\!\!1\!\!-\!\!\!\! \prod_{u_j \in \mathcal{N}}\!\!\!\!\left(\!1\!\!-\! c_{u_j, f_a}p_{u_i,u_j}\!\right)\!\!\!\right)
\end{align}
subject to
\begin{align}
\sum_{f_a \in \mathcal{F}}c_{u_i,f_a} \leq S, \hspace{1cm} \forall u_i \in \mathcal{N},
\end{align}
\begin{align}
 c_{u_j,f_a} \in \left\{0,1\right\}, \hspace{1cm} \forall u_j \in \mathcal{N}, f_a \in \mathcal{F},
\end{align}
\end{subequations}

\noindent where $c_{u_j,f_a}$ is defined in~\eqref{eq:c_u}. Here, the  normalization constant $m$ is dropped from \eqref{eq:ofp} since it does not depend on $\mathcal{C}^{N\times F}$. The term $1- \prod_{u_j \in \mathcal{N}}\left(1- c_{u_j, f_a}p_{u_i,u_j}\right)$ in ~\eqref{eq:ofp} indicates the hit probability $H_{u_i,f_a}$ in~\eqref{eq:c} to retrieve the content file $f_a\in \mathcal{F}$ via D2D single-hop communications, which thus corresponds to the probability that the requesting user $u_i$ comes in contact with another user $u_j$ storing the desired content. This implies that $\prod_{u_j \in \mathcal{N}}\left(1- c_{u_j, f_a}p_{u_i,u_j}\right)$ is the probability of the event that there is no user storing the content $f_a$ that user $u_i$ has requested within the distance $\mathcal{D}$. If user $u_j\in \mathcal{N}$ is holding the content requested by the user that is within the communication distance $\mathcal{D}$ (i.e., $c_{u_j,f_a}=1$), then the D2D communication can be initiated.  Intuitively, increasing $\mathcal{D}$ leads to an increment of the user density within $\mathcal{D}$, which yields an enhancement of $H_{avg}$. This due to the fact that the contact probability $p_{u_i,u_j}$ in~\eqref{eq:ofp}, which denotes the probability that two users $u_i$ and $u_j$ are within the communication distance $\mathcal{D}$ increases with the user density within $\mathcal{D}$, thus resulting in the improvement on the average hit ratio $H_{avg}$.

\begin{remark}~\label{R:outage} Note that the main focus of the paper is to design the optimal caching strategies that maximize the hit ratio. Thus, the optimal solution to the problem~\eqref{eq:opp} may not guarantee the optimality in terms of other performance metrics. However, we shall empirically validate the effectiveness of our caching framework with respect to the outage capacity in Section 6.2.
\end{remark}

As described in Section~\ref{section:21}, the model parameters such as user activity levels $P(u_i)$ and user file preferences $P(f_a|u_i)$ are assumed not to change during every time period $\mathcal{T}$ so that the caching strategy $\mathcal{C}^{N\times F}$ should be iteratively determined by the CP with updated parameters after every time period $\mathcal{T}$ during the off-peak time. Fig.~\ref{fig:cpc} is an illustration of the proposed framework composed of our parameter learning and content placement cycle, in which these two major tasks are repeated for every period with $\mathcal{T}$ time slots. First, the CP learns (and updates) the model parameters such as user activity levels and user file preferences after predicting missing information (e.g., missing ratings in recommender systems) of each user based on the data collected for the past $\mathcal{T}$ time slots. Second, by using the predicted parameters, the CP determines the caching strategy $\mathcal{C}^{N\times F}$ and initiates the content placement to fill the cache of each user.

\begin{figure}[t]
\centering
 \includegraphics[height=7.5cm]{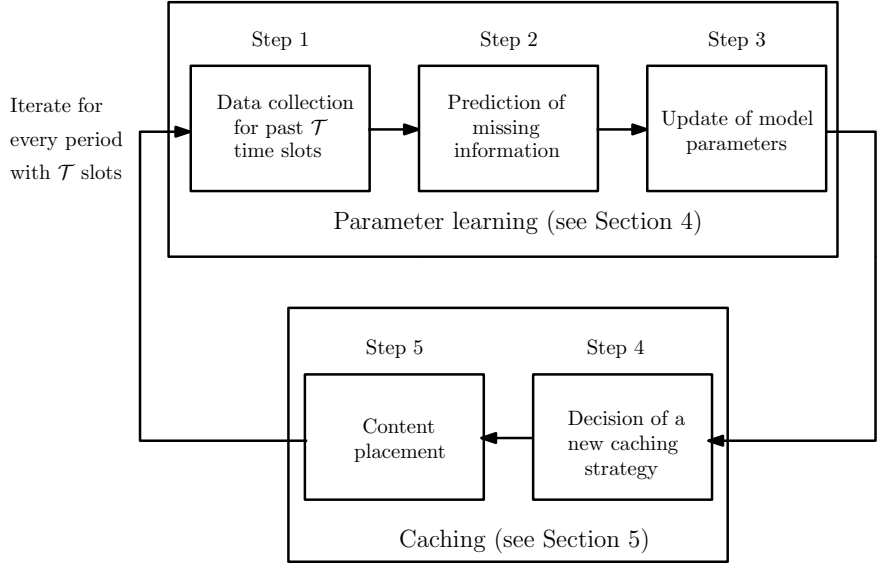}
\caption{The schematic overview of our parameter learning and  content placement cycle.}
\label{fig:cpc}
\end{figure}

\subsection{Baseline Strategy}~\label{section:33}
In this subsection, we describe a baseline caching strategy for comparison. While the focus of our study is on analyzing the impact of personalized file preferences in a content-centric mobile network, for our baseline, we use the case adopted in many previous studies~\cite{alfano,anh,milad,malik} where all users have the \textit{common} file preferences (i.e., $P(f_a|u_i)=P(f_a|u_j)= P(f_a)$), where $P(f_a)$ is the probability that a file $f_a$ is requested (i.e., the file popularity).

Now, from the caching constraints in \eqref{eq:Cons1} and \eqref{eq:Cons2}, the optimal cache allocation strategy $\mathcal{C}^{N\times F}$ for the content-centric mobile network employing the common file preferences can thus be the solution to the following optimization problem:

\begin{subequations} \label{eq:opb}
\begin{align} \label{eq:ofb}
\! \textbf{[P2]} \underset{\mathcal{C}^{N\times F}}{\max} \!\!\! \sum_{u_i \in \mathcal{N}} \!\! \sum_{f_a \in \mathcal{F}} P(u_i) P(f_a)\!\!\left(\!1\!\!- \!\!\!\!\prod_{u_j \in \mathcal{N}}\!\!\!\!\left(\!1\!\!-\! c_{u_j, f_a}p_{u_i,u_j}\!\right)\!\!\right)
\end{align}
subject to
\begin{align}
\sum_{f_a \in \mathcal{F}}c_{u_i,f_a} \leq S, \hspace{1cm} \forall u_i \in \mathcal{N},
\end{align}
\begin{align}
 c_{u_j,f_a} \in \left\{0,1\right\}, \hspace{1cm} \forall u_j \in \mathcal{N}, f_a \in \mathcal{F},
\end{align}
\end{subequations}

\noindent where $p_{u_i,u_j}$ is the contact probability that two users $u_i,u_j\in\mathcal{N}$ are within the communication distance $\mathcal{D}$ and $c_{u_j,f_a}$ is defined in~\eqref{eq:c_u}.

In practice, since users have different content file preferences, incorporating the knowledge of personalized file preferences into the cache allocation problem is expected to bring a substantial performance gain. This implies that content placement according to the caching strategy via the solution to the optimization problem in \eqref{eq:opb} will not perform better than the one obtained by solving the optimization problem in \eqref{eq:opp}, which will be empirically demonstrated in Section~\ref{section:6}.

\subsection{Problem Reformulation Using Submodular Properties}~\label{section:34}
In this subsection, we show that each of the optimization problems in~\eqref{eq:opp} and~\eqref{eq:opb} can be formulated as optimization (i.e., maximization) of a submodular function subject to matroid constraints. Since the problems in~\eqref{eq:opp} and~\eqref{eq:opb} are NP-hard~\cite{hit} and thus solving them for large system parameters such as $N$ and $F$ may not be tractable, a common way to solve these problems is the use of computationally efficient greedy algorithms. The computational efficiency comes at the cost of performance degradation. In this context, the problem reformulation enables us to obtain an approximate solution within provable gaps from greedy algorithms. Implementation details of greedy algorithms and the corresponding level of approximation that they can achieve compared to the optimal solution  will be specified in Section~\ref{section:5}.

First, we begin with formulating an equivalent optimization problem in~\eqref{eq:opp} similarly as in~\cite{hit}. Let $\mathcal{W} = \mathcal{N} \times \mathcal{F}$ denote a ground set, where $(u_i, f_a) \in \mathcal{W}$ represents the configuration that user $u_i \in \mathcal{N}$ caches file $f_a \in \mathcal{F}$ and the ground set $\mathcal{W}$ contains all possible configurations of file placement. Then, we aim to find the subset of $\mathcal{W}$
that maximizes the average hit ratio $H_{avg}$ under the constraint~\eqref{eq:Cons1}. We define a discrete set function $G : 2^{\mathcal{W}}\rightarrow \mathbb{R}$ on subsets $\mathcal{V}$ of the ground set $\mathcal{W}$ as follows:
\begin{align} \label{eq:rhrp}
G(\mathcal{V})\! \triangleq\!  1-\!\!\!\!\!\!\!\! \sum_{(u_i,f_a) \in \mathcal{V}}\!\!\!\!\!\!\! P(u_i) P(f_a|u_i) \!\!\left(\prod_{(u_j,f_a) \in \mathcal{V}\cap\mathcal{W}^{(f_a)}}\!\!\!\!\!\!\!\!\!\!\!\!\!\!\!\!\left(1- p_{u_i,u_j}\right)\!\!\right),
\end{align}
\noindent where $\mathcal{W}^{(f_a)}=\left\{(u_i, f_a): u_i \in\mathcal{N}\right\}$ indicates the possible placement configuration for  content file $f_a$ and $G(\mathcal{V})$ is a discrete set function, which gives us the average hit ratio for given subsets $\mathcal{V}$ (i.e., the caching strategy) of the ground set $\mathcal{W}$, when personalized file preferences $P(f_a|u_i)$ of users are available . Then, from~\eqref{eq:rhrp}, we are capable of establishing the following equivalent optimization problem to the original one in~\eqref{eq:opp}:
\begin{subequations} \label{eq:ropp}
\begin{align} \label{eq:rofp}
\textbf{[P3]}   \quad\quad\quad \quad \underset{\mathcal{V} \in \mathcal{J}}{\max} \hspace{0.1cm}   G(\mathcal{V})  \quad\quad\quad\quad
\end{align}
subject to
\begin{align} \label{eq:rCons1}
\left|\mathcal{V}\cap\mathcal{W}_{(u_i)}\right| \leq S \hspace{1cm} \forall \mathcal{V}\in \mathcal{J}, u_i\in \mathcal{N},
\end{align}
\end{subequations}
where $\mathcal{J} \in 2^{\mathcal{W}}$ is a family of feasible sets that satisfies the cache size constraint~\eqref{eq:rCons1} and $\mathcal{W}_{(u_i)}=\left\{(u_i, f_a): f_a \in\mathcal{F}\right\}$ indicates the possible cache placement configuration for user $u_i$'s cache. Note that each user's cache size constraint in~\eqref{eq:rCons1} is a matroid constraint corresponding to~\eqref{eq:Cons1}. We are now ready to show the following submodular property.
\begin{lemma}\label{le:eopp}
The function $G(\mathcal{V})$ in~\eqref{eq:rhrp} is a monotone non-decreasing submodular function as adding one more configuration to given configuration subsets (i.e., $(u_i, f_a) \cup \mathcal{V}$) leads to an increment of $G(\mathcal{V})$.
\end{lemma}
\begin{proof} The proof essentially follows the same steps as those in~\cite[Lemma 1]{hit} by just replacing $1/M$ (i.e., the term assuming the same activity level for all users) by the probability $P(u_i)$ with which a user $u_i$ generates a request (i.e., heterogeneous activity levels of users).
\end{proof}

Next, we move on to the baseline strategy where all users have the common file preferences (i.e., $P(f_a|u_i)=P(f_a|u_j)=P(f_a))$. According to the above problem reformulation argument, the following equivalent optimization problem to the original one in~\eqref{eq:opb} can be established:
\begin{subequations} \label{eq:ropb}
\begin{align} \label{eq:rofb}
\textbf{[P4]}   \quad\quad \quad\quad\quad \quad \underset{\mathcal{V} \in \mathcal{J}}{\max} \hspace{0.1cm}  G(\mathcal{V})  \quad\quad\quad\quad\quad 
\end{align}
subject to
\begin{align} \label{eq:rCons1b}
\left|\mathcal{V}\cap\mathcal{W}_{(u_i)}\right| \leq S \hspace{1cm} \forall \mathcal{V}\in \mathcal{J}, u_i\in \mathcal{N},
\end{align}
\end{subequations}
where $G(\mathcal{V})$ for the baseline strategy is given by
\begin{align} \label{eq:rhrb}
G(\mathcal{V}) \triangleq   1- \!\!\!\!\!\!\!\!\sum_{(u_i,f_a) \in \mathcal{V}} \!\!\!\!\!\!P(u_i) P(f_a)\!\! \left(\prod_{(u_j,f_a) \in \mathcal{V}\cap\mathcal{W}^{(f_a)}}\!\!\!\!\!\!\!\!\!\!\!\!\!\!\!\!\left(1- p_{u_i,u_j}\right)\right).
\end{align}
\begin{lemma}\label{le:eopb}
The function $G(\mathcal{V})$ in~\eqref{eq:rhrb} is a monotone non-decreasing submodular function as adding one more configuration to given configuration subsets (i.e., $(u_i, f_a) \cup \mathcal{V}$) leads to an increment of $G(\mathcal{V})$.
\end{lemma}

We shall design caching algorithms based on greedy approaches in terms of maximizing a submodular function subject to the matroid constraints in Section~\ref{section:5}.

\section{Parameter Learning}~\label{section:4} 

Thus far, we have characterized the caching optimization problems in Section~\ref{section:3} under the assumption that the key model parameters such as activity levels $P(u_j)$, file preferences $P(f_a|u_j)$, and file popularity $P(f_a)$ are available. In this section, we focus on the parameter learning phase of the iterative cycle illustrated in Fig.~\ref{fig:cpc} in order to completely solve the caching optimization problems in Section~\ref{section:3}. We describe a method to learn these three parameters based on the user rating behaviors, i.e., the rating matrix containing the ratings of $N$ users to $F$ content files (items). We first present a CF-based inference approach to predicting missing ratings for acquiring a complete rating matrix, denoted by ${\bf R}_c^{N\times F}$, from the fact that users tend to give ratings to only few items. We then show how to compute three parameters using the complete rating matrix. Note that since there is a little understanding of how to apply such rating prediction to wireless systems including cache-enabled networks, we adopt a simple CF-based approach in our study (rather than more sophisticated inference methods), which is sufficient to guarantee satisfactory performance improvement over the baseline strategy (refer to Section 6.2).

\subsection{Predicting Missing Ratings }~\label{section:41}
Collaborative filtering (CF) is one of prevalent techniques widely used in recommender systems mostly to predict ratings of unrated items or to suggest the most favorable top-$N$ items for a user by analyzing users' past behaviors on items. In CF, user behaviors are presented along with \textit{ratings} (explicit) or \textit{clicks} (implicit) on content files.  In our study, we focus on predicting the missing ratings using the user--item rating matrix based on the explicit feedback and each rating reflects a user {\em preference} for a certain item. The CF-based inference is built based upon the basic assumptions that (i) two users tend to have similar behavior patterns in the future when the two users share similar behavior patterns in the past; and (ii) the user behavior patterns are consistent over time.
The CF-based techniques are generally classified into memory-based approaches including user-based CF~\cite{ubcf} and item-based CF~\cite{ibcf} and model-based approaches such as singular value decomposition (SVD)-based CF~\cite{svd} and autoencoder-based CF~\cite{auto}. In CF, the missing ratings are predicted to obtain the complete rating matrix ${\bf R}_c^{N\times F}$ containing the ratings of $N$ users to $F$ files (items).

In our study, for ease of explanation with intuition, we adopt the memory-based CF approaches that perform predictions based on the similarity among users or items ~\cite{ubcf, ibcf}.\footnote{Since our parameter learning framework is CF-agnostic, model-based approaches and deep-learning approaches can also be adopted with a slight modification.} There exist two memory-based CF techniques: user-based CF~\cite{ubcf} and item-based CF~\cite{ibcf}.

In user-based CF, the missing rating $r_{k,a}$ of user $u_k \in \mathcal{N}$ for content file $f_a \in \mathcal{F}$ is predicted using the ratings given by the users similar to the target user $u_k \in \mathcal{N}$, and the similarity between the users is measured by referring to their ratings to other items.
Let $\mathcal{U}_{k,a}$ be the set of users who have rated content file $f_a \in \mathcal{F}$ and has similar rating behaviors to the target user $u_k \in \mathcal{N}$.
Then, the predicted rating $\hat{r}_{k,a}$ of user $u_k \in \mathcal{N}$ to content file $f_a \in \mathcal{F}$ is given by
\begin{align} \label{eq:pre_ru}
 \hat{r}_{k,a} = b_{k,a}+ \frac{\sum_{i \in \mathcal{U}_{k,a}} w_{k,i}(r_{i,a}- b_{i,a})}{\sum_{i \in \mathcal{U}_{k,a}} w_{k,i}},
\end{align}
where $r_{i,a}$ denotes the rating given by user $u_i \in \mathcal{N}$ to file $f_a \in \mathcal{F}$, $b_{i,a}$ denotes the biased rating value for $u_i \in \mathcal{N}$ to file $f_a \in \mathcal{F}$, and $w_{k,i}$ denotes the similarity weight between users $u_k,u_i \in \mathcal{N}$.

In item-based CF, the missing rating $r_{k,a}$ of user $u_k\in \mathcal{N}$ for the target content file $f_a \in \mathcal{F}$ is predicted using the ratings to the items similar to the target item $f_a \in \mathcal{F}$, and the similarity between the items is measured by referring to the ratings given by users. Let $\mathcal{I}_{a,k}$ be the set of items that have been rated by user $u_k\in \mathcal{N}$ and have similar rating behaviors to the target file $f_a \in \mathcal{F}$. Then, the predicted rating $\hat{r}_{k,a}$ of user $u_k \in \mathcal{N}$ to content file $f_a \in \mathcal{F}$ is given by
\begin{align} \label{eq:pre_ri}
 \hat{r}_{k,a} = b_{k,a}+ \frac{\sum_{j \in \mathcal{I}_{a,k}} w_{a,j}(r_{k,j}- b_{k,j})}{\sum_{j \in \mathcal{I}_{a,k}} w_{a,j}},
\end{align}
where $w_{a,j}$ is the similarity weight between the content files $f_a, f_j \in \mathcal{F}$.

For the user-based CF approach, the steps to predict each missing rating to a content file $f_a \in \mathcal{F}$ by user $u_j \in \mathcal{N}$ are summarized below:
\begin{enumerate}
\item First, we measure the similarity between users (files in case of item-based CF) using the available rating values. The popular similarity measures include the Pearson correlation, cosine similarity, and Euclidean distance. As an example, the similarity $w_{k,i}$ between two users $u_k,u_i\in\mathcal{N}$ using the Pearson correlation is given by
	\begin{align}\label{eq:sim}
	 w_{k,i}=  \frac{\sum_{j \in \mathcal{J}_{k,i}} (r_{k,j}- \bar{r}_{k} )\cdot (r_{i,j}-\bar{r}_{i}) }{\sqrt{ \sum_{j \in \mathcal{J}_{k,i}} (r_{k,j}-\bar{r}_{k})^2  \cdot \sum_{j \in \mathcal{J}_{k,i}} (r_{i,j}-\bar{r}_{i})^2}} 
	\end{align}
	where $\mathcal{J}_{k,i}$ is the set of items that have been rated by both users $u_k, u_i \in \mathcal{N}$ and $\bar{r}_{k}$ is the mean value of the ratings given by user $u_k \in \mathcal{N}$.
\item Next, we use the similarity values to identify the most similar users (i.e., $\mathcal{U}_{k,a}$) to the target user $u_k \in \mathcal{N}$ (the most similar files (i.e., $\mathcal{I}_{a,k}$) to the target file $f_a \in \mathcal{F}$ in case of item-based CF), where either the top $K$ users (files in case of item-based CF) or only the users (files in case of item-based CF) whose similarities are larger than a pre-defined threshold are taken into account.
\item Finally, we calculate the missing rating of user $u_j\in \mathcal{N}$ for the content file $f_a \in \mathcal{F}$ using \eqref{eq:pre_ru} (\eqref{eq:pre_ri} in case of item-based CF).
\end{enumerate}

To further improve the prediction accuracy, it is also possible to apply both matrix factorization techniques such as SVD-based CF~\cite{svd} to learn the hidden user patterns based on the observed ratings and recently emerging deep learning techniques such as autoencoder-based CF~\cite{auto} into our prediction framework. In the next subsection, we explain how to calculate the model parameters based on the complete rating matrix $\textbf{R}_c^{N\times F}$.

\subsection{Parameter Calculation}~\label{section:42}
 In this section, we calculate user activity levels $P(u_k)$, user file preferences $P(f_a|u_k)$, and file popularity $P(f_a)$ using the rating matrix $\textbf{R}_c^{N\times F}$ collected in $\mathcal{T}$ time slots. Let $r_{k,a}$ denote the rating given by user $u_k \in \mathcal{N}$ to file $f_a \in \mathcal{F}$ and also $n_{k}$ denote the total number of ratings given (not predicted) by user $u_k \in \mathcal{N}$. Suppose that each user gives ratings to all the content files that he/she requests, which implies that $n_k$, the number of ratings given by user $u_k$,  is identical to the number of requests generated by the user. Then, the user activity level $P(u_k)$ can be expressed as
\begin{align} \label{eq:al}
 P(u_k) = \frac{n_k}{\mathcal{T}}.
\end{align}
In addition, the file preferences $P(f_a|u_k)$ of each user and the file popularity $P(f_a)$ are given by
\begin{align} \label{eq:fp}
 P(f_a|u_k) = \frac{r_{k,a}}{\sum_{i=1}^F r_{k,i}}
\end{align}
and
\begin{align} \label{eq:pop}
 P(f_a) =  \frac{\sum_{k=1}^N r_{k,a}}{\sum_{k=1}^N\sum_{i=1}^F r_{k,i}},
\end{align}
respectively.

\begin{figure}[t]
\centering
 \includegraphics[height=7.5cm]{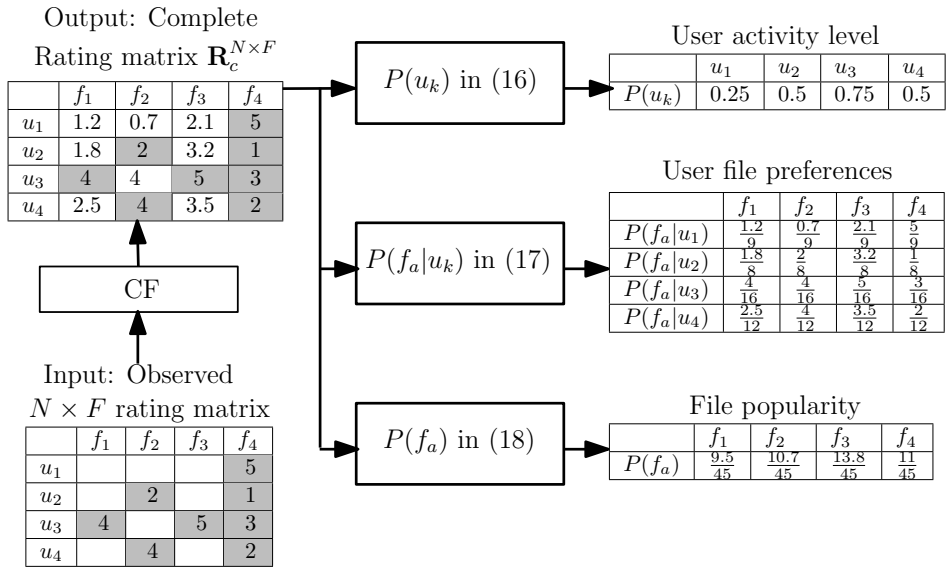}
\caption{Block diagram of parameter learning with an example of $\mathcal{T}=N=F=4$.}
\label{fig:block}
\end{figure}

Fig. \ref{fig:block} illustrates the parameter learning procedure. Consider an example of $N=F=\mathcal{T}=4$ as shown in Fig. \ref{fig:block}. We start with the observed $4 \times 4$ rating matrix containing the missing ratings and let the matrix pass through one of CF algorithms for predicting the missing ratings. The output of the CF algorithm is a complete rating matrix $\textbf{R}_c^{4 \times4}$. Then, user activity levels can be calculated from~\eqref{eq:al}. For example, the probability that user $u_1$ generates a request is $P(u_1)=0.25$ as the total number of ratings given (not predicted) by user $u_1$ in $4$ time slots is $n_1=1$. Next, user file preferences can be calculated from~\eqref{eq:fp}. For example, the probability that user $u_2$ prefers file $f_2$ at the moment of requesting the content is $P(f_2|u_2)=2/8$ as the rating given by user $u_2$ to file $f_2$ is $r_{2,2}=2$ and the sum of all the ratings of user $u_2$ to all files is  $\sum_{i=1}^4 r_{2,i}=8$. Lastly, the file popularity can be calculated from~\eqref{eq:pop}. For example, the probability that file $f_4$ is preferred at the moment of requesting the content is $P(f_4)=11/45$ as the sum of ratings given by all users to file $f_4$ is $\sum_{k=1}^4 r_{k,4}=11$ and the sum of all the ratings given by all users to all four files is $\sum_{k=1}^4\sum_{i=1}^4 r_{k,i}=45$.

\begin{remark}~\label{R:0}
Let us discuss the difference between our study and~\cite{chen} in terms of the basic definitions in \eqref{eq:al}--\eqref{eq:pop}. In~\cite{chen}, the activity level of user $u_k\in\mathcal{N}$ is given by the ratio of the number of requests generated by $u_k$ to the total number of requests generated by all the users (i.e., $P(u_k)=\frac{n_k}{\sum_{j=1}^N n_j}$); whereas in our study, we measure how frequently each user has been active with content requests, where the activity level of one user is independent of other users' activities. Additionally, in~\cite{chen}, user file preferences $P(f_a|u_k)$ and file popularity $P(f_a)$ are given by $\frac{n_{k,a}}{ \sum_{i=1}^F n_{k,i}}$ and $\sum_{k=1}^N\frac{n_{k,a}}{ \sum_{\tilde{k}=1}^N \sum_{i=1}^F n_{\tilde{k},i}}$, respectively, where $n_{k,a} \geq 0$ is the cumulative number of requests from $u_k \in \mathcal{N}$ to content file $f_a \in \mathcal{F}$. On the other hand, in our study, we utilize the user ratings as a metric to learn the above model parameters instead of using the number of times that each user accesses a file. Although the frequency of accessing a file may reflect user file preferences, it is true only such types of content files that tend to be accessed repeatedly such as music videos. This assumption may not hold for other types of content files which do not need to be accessed more than once as a user tends to watch them only once. In this context, the user rating history would be proper for both types of content objects stated above and can reflect a better understanding of each user behavior.
\end{remark}

\section{Caching Algorithms }~\label{section:5}
In this section, we present the implementation details of several caching algorithms to solve the problems in Section~\ref{section:3} by focusing on the content placement phase of the iterative cycle illustrated in Fig.~\ref{fig:cpc}. We also analyze the computational complexity for each algorithm while showing the corresponding level of the approximation that our greedy approach can achieve compared to the optimal solution. We first show an approach via exhaustive search to obtaining the optimal solution to the problems in~\eqref{eq:ropp} and~\eqref{eq:ropb}. Moreover, we present two computationally efficient heuristic algorithms at the cost of slight performance degradation compared to the optimal approach.

\subsection{Optimal Approach}~\label{section:51}
This subsection deals with an exhaustive search method over all feasible caching strategies to obtain the optimal caching strategy by solving the problems in~\eqref{eq:ropp} and~\eqref{eq:ropb}, which is summarized in the Optimal Algorithm. For given input parameters such as the activity levels $P(u_j)$, the user file preferences $P(f_a|u_j)$ (the file popularity $P(f_a)$ for the baseline strategy), and the contact probability $p_{u_i,u_j}$, the algorithm starts with an empty set $\hat{\mathcal{V}}$, representing the optimal caching strategy, and zero cost (i.e., hit ratio). The algorithm runs for $\left|\mathcal{J}\right|$ iterations, where $\mathcal{J}$ is a family of feasible subsets (i.e, caching strategies) that satisfies the storage capacity constraint~\eqref{eq:rCons1} and $\left|\mathcal{J}\right| = \mathcal{O}(2^{NF})$. At each iteration, the hit ratio $G(\mathcal{V})$ of a feasible subset $\mathcal{V} \in \mathcal{J}$ is computed according to~\eqref{eq:rhrp} (\eqref{eq:rhrb} for the baseline strategy) and the subset $\mathcal{V} \in \mathcal{J}$ with the maximum hit ratio $G(\mathcal{V})$ is selected as an optimal caching strategy $\hat{\mathcal{V}}$. 

Next, we turn to analyzing the computational complexity of the Optimal Algorithm.
\begin{remark}\label{R:1}
The time complexity of the Optimal Algorithm is  $\mathcal{O}(NF 2^{NF})$ as the algorithm runs for $\mathcal{O}(2^{NF})$ iterations (i.e., the total number of subsets representing all the possible configurations that user $u_i$ caches file $f_a$ for $\forall \ u_i \in \mathcal{N}$ and $f_a \in \mathcal{F}$ is $2^{NF}$). In each iteration, the evaluation of $G(\mathcal{V})$ takes $NF$ basic operations according to~\eqref{eq:rhrp} (\eqref{eq:rhrb} for the baseline strategy). This finally gives us the time complexity of $\mathcal{O}(NF2^{NF})$.
\end{remark}
The high complexity of this algorithm makes it impossible to find the optimal solution within polynomial time even for small system parameters such as $N$, $F$, and $S$. To overcome this issue, we would like to introduce computationally efficient algorithms in the next subsection.

\begin{table}[t]
\centering
\begin{tabular}{p{15cm}}
\hline
Optimal Algorithm \\
\hline
\textbf{Input:} $P(u_j\!)$, $P(f_a|u_j\!)$, $P(f_a\!)$, and $p_{u_i,u_j}$, \!$\forall\  u_j \! \in \! \mathcal{N}, f_a \! \in \! \mathcal{F}$ \\
\textbf{Output:} $\hat{\mathcal{V}}$ \\
\textbf{Initialization:} $\hat{\mathcal{V}} \leftarrow \emptyset$; $cost \leftarrow 0$ \\
01: \textbf{for} each $\mathcal{V} \in \mathcal{J}$  \textbf{do}\\
02: \ \ \ Calculate $G(\mathcal{V})$ \\
03: \ \ \ \textbf{if} $cost < G(\mathcal{V})$ \textbf{then}\\
04:	\ \ \ \ \ $\hat{\mathcal{V}} \leftarrow \mathcal{V}$ \\
05:	\ \ \ \ \ $cost \leftarrow G(\mathcal{V})$ \\
06:	\ \ \ \textbf{end if} \\
07:	\textbf{end for} \\
\hline
\end{tabular}
\end{table}

\begin{table}[t]
\centering
\begin{tabular}{p{15cm}}
\hline
Greedy Algorithm ({\bf G1} Algorithm) \\
\hline
\textbf{Input:} $P(u_j\!)$, $P(f_a|u_j\!)$, $P(f_a\!)$, and $p_{u_i,u_j}$, $\forall u_j \! \in \!\mathcal{N}, f_a \in \!\mathcal{F}$ \\
\textbf{Output:} $\mathcal{V}$ \\
\textbf{Initialization:} $\mathcal{V} \leftarrow \emptyset$;  $\mathcal{F}_{u_i} \leftarrow \mathcal{F}$, $\forall\  u_i \in \mathcal{N}$ ; $\mathcal{U}\leftarrow \mathcal{N}$ \\
01: \textbf{for} $k$ \textbf{from} 1 \textbf{to} $S\times N$ \textbf{do} \\
02: \ \ $[u_{\hat{i}},\! f_{\hat{a}}]\!\! \leftarrow \!\! \arg \max_{u_i \in \mathcal{U}, f_a \in \mathcal{F}_{u_i} }\!\! \left(G(\mathcal{V}\cup (u_i, f_a))\!\! -\!\!G(\mathcal{V}) \right)$ \\
03: \ \ $\mathcal{V} \leftarrow \mathcal{V}\cup (u_{\hat{i}}, f_{\hat{a}})$ \\
04: \ \ $\mathcal{F}_{u_{\hat{i}}} \leftarrow \mathcal{F}_{u_{\hat{i}}} \symbol{92}  f_{\hat{a}}$ \\
05: \ \ \textbf{if} $\left| \mathcal{F}_{u_{\hat{i}}} \right| = F-S$ \textbf{then} \\
06: \ \ \ \ $\mathcal{U} \leftarrow \mathcal{U} \symbol{92} u_{\hat{i}}$ \\
07: \ \ \textbf{end if} \\
08: \textbf{end for} \\
\hline
\end{tabular}
\end{table}

\subsection{Computationally Efficient Approaches}~\label{section:52}
In this subsection, we present two types of efficient algorithms along with analysis of their computational complexities.

\subsubsection{Greedy Algorithm ({\bf G1} Algorithm)}~\label{section:521}
We first present a greedy approach, which is summarized in the {\bf G1} Algorithm. For given input parameters $P(u_j)$, $P(f_a|u_j)$, $P(f_a)$, and $p_{u_i,u_j}$, the algorithm starts by initially setting $\mathcal{V}$, $\mathcal{F}_{u_i}$, and $\mathcal{U}$ to $\emptyset$, $\mathcal{F}$, $\mathcal{N}$, respectively, where $\mathcal{F}_{u_i}$ denotes the set of content files not cached by user $u_i \in\mathcal{N}$ and $\mathcal{U}$ denotes the set of users with residual cache space. Given the system parameters $N$, $F$, and $S$, the algorithm runs for $SN$ iterations. In each iteration, one feasible element $(u_{i}, f_{a}) \notin \mathcal{V}$ that satisfies the storage capacity constraint~\eqref{eq:rCons1} is added and the highest marginal increase in the hit ratio is offered if added to the set $\mathcal{V}$.

Now, we analyze the performance gap between the {\bf G1} Algorithm and the Optimal Algorithm in Section~\ref{section:51} as follows.
\begin{prop}\label{p:g1}
The caching strategy obtained from the {\bf G1} Algorithm achieves the hit ratio within the factor of $\frac{1}{2}$ of the optimum.
\end{prop}
\begin{proof}
The proof technique basically lies in the similar steps to those established in~\cite{greedy1} since the {\bf G1} Algorithm is designed based on the reformulated problem using submodular properties.
\end{proof}
Next, we turn to analyzing the computational complexity of the {\bf G1} Algorithm.
\begin{remark}\label{R:2}
The time complexity of {\bf G1} Algorithm is $\mathcal{O}(SN^3F^2)$ as the algorithm runs for $SN$ iterations (refer to lines 1--8). Each iteration involves evaluating the marginal value (i.e., $G\left(\mathcal{V}\cup (u_i, f_a)\right)$ $-$ $G(\mathcal{V})$ in line 2) of at most $NF$ elements, that is, the number of unassigned elements $(u_i, f_a)$ in each iteration is given by $\sum_{u_i \in \mathcal{U}} \left| \mathcal{F}_{u_i}\right|$, and each evaluation of $G(\mathcal{V})$ takes $NF$ basic operations. This finally gives us the time complexity of $\mathcal{O}(SN(NF)NF)=\mathcal{O}(SN^3F^2)$.
\end{remark}

\begin{remark}\label{R:3}
Note that there exist algorithms that can perform better than the {\bf G1} Algorithm at the cost of higher time complexity. In \cite{greedy2}, a greedy algorithm was designed in such a way that maximizes a general monotone submodular function subject to matroid constraints, while guaranteeing an approximate solution within the factor of $1-\frac{1}{e}$ of the optimum. The approximate solution is briefly described according to the following two steps. First, the combinatorial integer programming problem is replaced with a continuous problem by relaxing the integer value constraint, and then the solution to the continuous problem is found. Second, the pipage rounding method is applied to the solution to the continuous problem to obtain a feasible approximate solution. Although this algorithm guarantees a better approximate solution within the factor of $1-\frac{1}{e}$ compared to the {\bf G1} Algorithm, this comes at the cost of high computational complexity as the time complexity of this algorithm is $\mathcal{O}((NF)^8)$. The high time complexity of this algorithm makes it impractical for large-scale mobile networks.
\end{remark}

Since the complexity of the {\bf G1} Algorithm may be still high for large system parameters $N$, $F$, and $S$, we shall seek a possibility to further reduce the complexity with acceptable performance degradation. To this end, we design a faster heuristic algorithm that has almost comparable performance to that of the {\bf G1} Algorithm in the next subsection.

\subsubsection{Greedy-Like Algorithm ({\bf G2} Algorithm)}~\label{section:522}
We introduce a more computationally efficient approach designed in a greedy-like manner, which is summarized in the {\bf G2} Algorithm. For given input parameters $P(u_j)$, $P(f_a|u_j)$, $P(f_a)$, $p_{u_i,u_j}$, and the number of iterations, denoted by $l$, the algorithm starts by initially setting $\mathcal{V}$ , $\mathcal{F}_u$, and $\mathcal{U}$ to $\left\{\mathcal{A}_{u_1}\cup \mathcal{A}_{u_2} \cup \cdots \cup \mathcal{A}_{u_N}\right\}$, $\mathcal{F}\setminus \mathcal{A}_{u_i}$, and $\emptyset$, respectively, where $\mathcal{A}_{u_i}$ is the set of $S$ most preferred content files for user $u_i \in\mathcal{N}$, $\mathcal{F}_{u_i}$ denotes the set of content files not cached by user $u_i \in\mathcal{N}$, and $\mathcal{U}$ denotes the set of users with residual cache space. Given the system parameters $N$, $F$, and $S$, the algorithm runs for $l$ iterations, each of which is divided into two parts. The first part runs $N$ times, where at each time, an element $(u_{i}, f_{a}) \in \mathcal{V}$ is removed from the set $\mathcal{V}$ that offers the  lowest marginal decrease in the hit ratio if removed from the set $\mathcal{V}$. Then, the second part runs $N$ times, where at each time, one feasible element $(u_{i}, f_{a}) \notin \mathcal{V}$ that satisfies the storage capacity constraint~\eqref{eq:rCons1} is added and the highest marginal increase in the hit ratio is offered if added to the set $\mathcal{V}$.
\begin{table}[t]
\centering
\begin{tabular}{p{15cm}}
\hline
Greedy-Like Algorithm ({\bf G2} Algorithm) \\
\hline
\textbf{Input:} \! $l$, \!$P(u_j\!)$, \!$P(f_a|u_j\!)$, \!$P(f_a\!)$, and $p_{u_i,u_j}$, \!$\forall\ \!u_j\! \in\! \mathcal{N}, f_a \!\in \!\mathcal{F}$ \\
\textbf{Output:} $\mathcal{V}$ \\
\textbf{Initialization:} $\mathcal{V} \leftarrow$ $\left\{\mathcal{A}_{u_1}\cup \mathcal{A}_{u_2} \cup \cdots \cup \mathcal{A}_{u_N}\right\}$; $\mathcal{U}\leftarrow \emptyset$; $\mathcal{F}_{u_i} \leftarrow \mathcal{F} \symbol{92} \mathcal{A}_{u_i}$, $\forall\  u_i \in\mathcal{N}$  \\
01: \textbf{for} $k$ \textbf{from} 1 \textbf{to} $l$ \textbf{do} \\
02: \ \ \textbf{for} $j$ \textbf{from} 1 \textbf{to}  $N$ \textbf{do} \\
03: \ \ \ \ $[u_{\hat{i}}, f_{\hat{a}}] \!\! \leftarrow \!\! \arg \min_{(u_i,f_a) \in \mathcal{V} } \left(G(\mathcal{V}) -G(\mathcal{V} \symbol{92} (u_i, f_a)) \right)$ \\
04: \ \ \ \ $\mathcal{V} \leftarrow \mathcal{V}\symbol{92} (u_{\hat{i}}, f_{\hat{a}})$ \\
05: \ \ \ \ $\mathcal{F}_{u_{\hat{i}}} \leftarrow \mathcal{F}_{u_{\hat{i}}} \cup  f_{\hat{a}}$ \\
06: \ \ \ \ \textbf{if} $\left| \mathcal{F}_{u_{\hat{i}}} \right| > F-S$ \textbf{and} $u_{\hat{i}}\notin \mathcal{U}$ \textbf{then} \\
07: \ \ \ \ \ \ $\mathcal{U} \leftarrow \mathcal{U} \cup u_{\hat{i}}$ \\
08:\ \ \ \ \ \textbf{end if} \\
09:\ \ \ \textbf{end for} \\
10:\ \ \ \textbf{for} $q$ \textbf{from} 1 \textbf{to}  $N$ \textbf{do} \\
11:\ \ \ \ \ $[u_{\hat{i}}, f_{\hat{a}}]\!\! \leftarrow\!\! \arg \max_{u_i \in \mathcal{U}, f_a \in \mathcal{F}_{u_i} } \left(G(\mathcal{V}\cup (u_i, f_a)) -G(\mathcal{V}) \right)$ \\
12:\ \ \ \ \ $\mathcal{V} \leftarrow \mathcal{V}\cup (u_{\hat{i}}, f_{\hat{a}})$ \\
13:\ \ \ \ \ $\mathcal{F}_{u_{\hat{i}}} \leftarrow \mathcal{F}_{u_{\hat{i}}} \symbol{92}  f_{\hat{a}}$ \\
14:\ \ \ \ \ \textbf{if} $\left| \mathcal{F}_{u_{\hat{i}}} \right| = F-S$ \textbf{then} \\
15:\ \ \ \ \ \ \ $\mathcal{U} \leftarrow \mathcal{U} \symbol{92} u_{\hat{i}}$ \\
16:\ \ \ \ \ \textbf{end if} \\
17:\ \ \ \textbf{end for} \\
18: \textbf{end for} \\

\hline
\end{tabular}
\end{table}

Next, let us analyze the computational complexity of the {\bf G2} Algorithm.
\begin{remark}\label{R:4}
The complexity of the {\bf G2} Algorithm is $\mathcal{O}(N^3F^2)$ as the algorithm runs for $l$ iterations (refer to lines 1--8). The first part of each iteration involves evaluating the marginal value (i.e., $G\left(\mathcal{V}\right) - G(\mathcal{V}\symbol{92}(u_i, f_a))$  in line 3) of $SN$ elements $N$ times. The second part of each iteration involves evaluating the marginal value  (i.e., $G\left(\mathcal{V}\cup (u_i, f_a)\right) - G(\mathcal{V})$ in line 11) of at most $NF$ elements, that is, $\sum_{u_i \in \mathcal{U}} \left| \mathcal{F}_{u_i} \right|$, $N$ times. Since each evaluation of $G(\mathcal{V})$ takes $NF$ basic operations, this finally gives us the time complexity of $\mathcal{O}(lNF(SN^2+ N^2F)) =\mathcal{O}(N^3F^2)$. The {\bf G2} Algorithm is computationally faster than the {\bf G1} Algorithm as $l$ is just a constant that does not scale with other system parameters.
\end{remark}

\subsubsection{Comparison With the Prior Work in \cite{femto}}~\label{section:523}
In this subsection, we compare our caching method with the so-called {\em FemtoCaching} system in \cite{femto}. Our study basically follows the same arguments as those in \cite{femto} in the sense of formulating a problem as the maximization of a submodular function and designing a greedy algorithm with an approximate solution within a provable gap. However, there are fundamental differences between two, which are summarized as follows.
\begin{itemize}
\item {\bf Network model}: Femto-cell static network and D2D mobile network models are adopted in \cite{femto} and our study, respectively. Moreover, a fixed network topology was assumed in \cite{femto} by using a single instance of connectivity of helpers and users, while dynamic network topologies are assumed exploiting the statistics (.e.g., the contact probability) of mobile users' connectivity under the protocol model in our study.
\item {\bf Caching model}: Content files are assumed to be cached at static helper nodes (i.e., dedicated small-cell base stations) and at cache-enabled mobile users in \cite{femto} and our work, respectively.
\item {\bf Performance metric}: The average per-bit downloading delay was used in \cite{femto}, whereas the average hit ratio and outage capacity are adopted in our study.
\item {\bf Model parameters}: In \cite{femto}, analysis was conducted based on the uniform user activity levels (i.e., the case where users are always requesting content files) and the common file popularity. In contrast, in our study, a heterogeneous activity level of each user is employed from the fact that it is necessary to distinguish high and low activity users in order to design an effective caching strategy due to the limited cache size at each mobile user; and the personalized file preferences are employed along with their inference using the CF-based approach.
\end{itemize}

\section{Experimental Evaluation}~\label{section:6}
In this section, we perform data-intensive simulations with finite system parameters $N$, $F$, $S$, and $\mathcal{D}$ for obtaining numerical solutions to the optimization problems in \eqref{eq:ropp} and \eqref{eq:ropb}. We compare the results of the proposed caching strategies in Section~\ref{section:32} with the baseline strategies in Section~\ref{section:33} to analyze the impact of personalized file preferences in our content-centric mobile network. We first elaborate on our experimental setup including the real-world dataset and network settings. Then, we present the numerical results according to various system parameters. 

\subsection{Experimental Setup}~\label{section:61}

\subsubsection{Dataset and Pre-Filtering}~\label{section:611}
We use MovieLens 1M\footnote{http://grouplens.org/datasets/movielens.} for our experiments, which is one of publicly available real-world datasets for evaluating movie recommender systems\cite{ibcf, ubcf,svd,sparsity}. The dataset includes 3,952 items (movies), 6,040 users, and 1,000,209 ratings, where the ratings take integer values ranging from 1 (i.e., the worst) to 5 (i.e., the best). 
 
 Note that the sparsity (i.e., the ratio of the number of missing cells in a matrix to the total number of cells) of the $6,040\times 3,952$ rating matrix obtained from the original dataset is 95.8\%, which is very high and often causes performance degradation. To resolve the sparsity issue, there have been a variety of studies (see~\cite{sparsity,9ofwww, 10ofwww} and references therein). One popular solution to the data sparsity problem is the use of {\em data imputation}~\cite{sparsity, 9ofwww, 10ofwww}. The zero injection method was proposed in~\cite{sparsity} by defining the pre-use preference as judgment of each user on items before using or purchasing them. Due to the belief that uninteresting items corresponding to the ones having low ranks of inferred pre-use preference scores would not satisfy the belonging users even if recommended, zeros are given to all of the uninteresting items in a rating matrix. Meanwhile, the pureSVD method~\cite{9ofwww} and the allRank method~\cite{10ofwww} that assign zeros and twos, respectively, to all missing cells in a rating matrix were introduced. Even if such data imputation techniques are known to significantly improve the prediction accuracy, we do not employ them in our experiments since solving the data sparsity problem is not our primary focus. 

As an alternative, we use a na\"{\i}ve {\em pre-filtering} method to obtain a less sparse dataset before applying the CF-based rating prediction approach since our main focus is not on improving the performance of recommender systems. Our pre-filtering method is given below:
		\begin{enumerate}
			\item First, among 6,040 users, we select users who have rated at least 400 items (i.e., movies) out of 3,952 items. This gives us a reduced $333\times 3,952$ rating matrix.
			\item From the reduced dataset, we select only the items that have been rated by at least 180 users. This finally gives us a $333 \times 261$ rating matrix with sparsity of 45.5\%.
		\end{enumerate}

 The statistics for the original and pre-filtered datasets are summarized in Table~\ref{table:data}. We then predict all the missing ratings in the pre-filtered rating matrix via the CF-based approach. In our experiments, we randomly select $N$ users and $F$ items out of 333 users and 261 items, respectively, in the pre-filtered rating matrix. We use the Pearson correlation in~\eqref{eq:sim} to evaluate the similarity.
 

\begin{table}[t]
\centering
\caption{Statistics of original and pre-filtered datasets }
\label{table:data}
\begin{tabular}{@{}llll@{}}
\hline
   & \begin{tabular}[c]{@{}l@{}}Original dataset\\ (MovieLens 1M)\end{tabular} & \begin{tabular}[c]{@{}l@{}}Pre-filtered\\ dataset\end{tabular} &\\
\hline
Number of users & 6,040    & 333       \\
Number of items & 3,952    & 261   \\
Total number of cells & 23,870,080    & 86,913    \\
Number of missing cells & 22,869,871   &39,541  \\
Sparsity & 95.8\%  & 45.5\% \\
\hline
\end{tabular}
\end{table}

\subsubsection{Network Setup}~\label{section:612}
Now, let us turn to describing our network setup for experiments. We assume that each user moves independently according to the random walk mobility model~\cite{alfano, malik} in a square torus network of unit area. In the mobility model, the position $x(t)$ of a user at time slot $t$ is updated by $x(t)= x(t-1)+ y_t$, where $y_t$ is a sequence of independent and identically distributed (i.i.d.) random variables that represent a user's flight vector. In order to calculate the contact probability $p_{u_i,u_j}$ in our experiments, we create a synthetic mobility trajectory database of $N$ users for $\mathcal{T}$ time slots according to the following steps:
\begin{itemize}
	\item At $t=1$, $N$ users are distributed independently and uniformly on a square network, which is divided into $s_N$ smaller square cells of equal size.
	\item Then at $t>1$, each user moves independently in any direction from its current location with a flight length uniformly and randomly selected from 0 to the size of each square cell, $\frac{1}{s_N}$.
\end{itemize}  
As addressed in Section~\ref{section:22}, a requesting user can successfully retrieve its desired content from his/her own cache or by communicating directly to other mobile users within the collaboration distance $\mathcal{D}$. For given $\mathcal{D}$, using the mobility trajectory database, it is possible to calculate the average contact probability $p_{u_i,u_j}$ at each time slot as the ratio of the time slots when two users $u_i$ and $u_j$ are within the distance $\mathcal{D}$ to the total number of time slots, $\mathcal{T}$. 

\subsubsection{Setup of System Parameters}~\label{section:613}
In order to analyze the impact of each system parameter on the performance of our proposed strategy, we first perform four types of experiments while evaluating the average hit ratio $H_{avg}$ according to different values of 1) the cache size $S$, 2) the number of users, $N$, 3) the library size $F$, and 4) the collaboration distance $\mathcal{D}$. Next, to investigate the effects of predicted ratings on the hit ratio performance, we conduct another experiment while showing $H_{avg}$ versus the masking percentage of the total given ratings. Finally, we have demonstrated that the proposed strategy is also effective with respect to the outage capacity $\Lambda$ by evaluating $\Lambda$ according to the collaboration distance $\mathcal{D}$.

\subsection{Numerical Results}~\label{section:62}
In this section, we perform numerical evaluation via intensive computer simulations using the pre-filtered MovieLens dataset. We evaluate the hit ratio for the following five algorithms: the optimal algorithm via exhaustive search ({\bf Opt-P1}), the {\bf G1} algorithm for solving~\eqref{eq:ropp} ({\bf G1-P3}), the {\bf G2} algorithm for solving~\eqref{eq:ropp} ({\bf G2-P3}), the {\bf G1} algorithm for solving~\eqref{eq:ropb} ({\bf G1-P4}), and the {\bf G2} algorithm for solving~\eqref{eq:ropb} ({\bf G2-P4}). In our experiments, the {\em user-based} CF technique is adopted as one of CF-based prediction approaches, but other CF techniques such as item-based CF and model-based approaches (e.g., SVD and autoencoder-based CF methods) can also be applied for predicting the ratings. The number of square cells is set to $s_N=40,000$ for simplicity.
\begin{figure}[t]
\centering
 \includegraphics[width=0.55\linewidth]{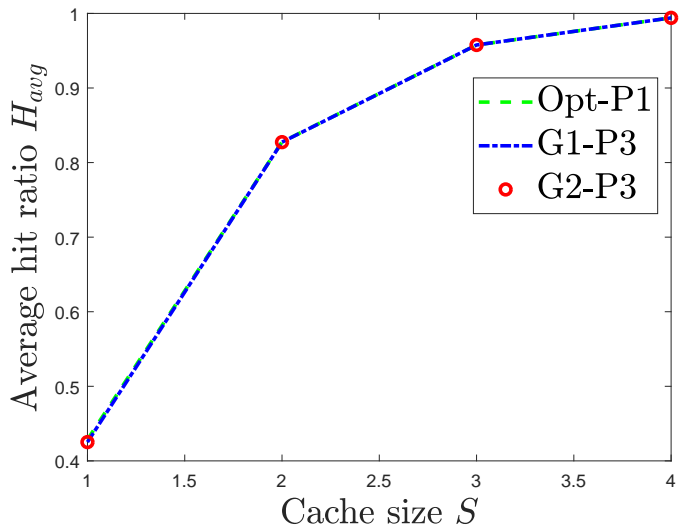}
\caption{The average hit ratio $H_{avg}$ versus the cache size $S$, where algorithms {\bf Opt-P1}, {\bf G1-P3}, and {\bf G2-P3} are compared.}
\label{fig:op}
\end{figure}

\subsubsection{Comparison Between the Optimal and Greedy Approaches}
In Fig.~\ref{fig:op}, we first show how the average hit ratio $H_{avg}$ behaves according to different values of the cache size $S$ to see how close performance of our two greedy algorithms {\bf G1-P3} and {\bf G2-P3} is to that of the optimal one {\bf Opt-P1}, where $N=F=5$ and $\mathcal{D}=0.1$. In this experiment, small values for system parameters $N$ and $F$ are chosen due to the extremely high complexity of {\bf Opt-P1} as $2^{25}$ searches over the cache placement are necessary even with $N=F=5$. In Fig.~\ref{fig:op}, it is obvious that the hit ratio grows monotonically with increasing $S$ since the contact probability becomes higher with increasing $S$. More interestingly, since the performance of {\bf G1-P3} and {\bf G2-P3} is quite comparable to that of {\bf Opt-P1}, using our scalable greedy strategies can be a good alternative  in large-scale mobile networks. In the subsequent experiments, we evaluate the performance of our four greedy algorithms in Section~\ref{section:5} and perform a comparative analysis among them.

\subsubsection{Comparison Between the Proposed and Baseline Approaches}
In Fig.~\ref{fig:wrt}a, the average hit ratio $H_{avg}$ versus the cache size $S$ is illustrated to see how much caching strategies employing the personalized file preferences are beneficial over the case with the common file preferences, where $N=30$, $F=200$, and $\mathcal{D}= 0.05$. As expected, $H_{avg}$ is enhanced with increasing $S$. We can observe that knowledge on the personalized file preferences brings a significant gain compared to the baseline strategy in which all users have the common file preferences (i.e., $P(f_a)= P(f_a|u_i)=P(f_a|u_j)$). The performance difference between two caching strategies becomes significant for a higher value of $S$ as depicted in the figure.

In Fig.~\ref{fig:wrt}b, we illustrate the average hit ratio $H_{avg}$ according to different values of the number of users, $N$, where $S=10$, $F=50$, and $\mathcal{D}= 0.1$. It is observed that the proposed strategy ({\bf G1-P3} and {\bf G2-P3}) consistently outperforms its counterpart scheme employing the common file preferences ({\bf G1-P4} and {\bf G2-P4}). Another interesting observation is that the performance gap between the two caching strategies tends to be reduced with increasing $N$ as depicted in the figure. This is because an increment of $N$ not only increases the overall storage capacity of the network but also the chance of finding a source user within $\mathcal{D}$ increases by virtue of an increased user density within the $\mathcal{D}$.

In Fig.~\ref{fig:wrt}c, the average hit ratio $H_{avg}$ versus the content library size $F$ is illustrated, where $S=15$, $N=30$, and $\mathcal{D}= 0.1$. As expected, $H_{avg}$ decreases with increasing $F$  since the storage capacity of the network does not increase with $F$. We also observe that the performance difference between the two caching strategies employing the personalized file preferences and the common file preferences is significant for almost all values of $F$.
\begin{figure}[t]
  \centering
  \subfigure[Cache size $S$]{\includegraphics[width=.49\linewidth]{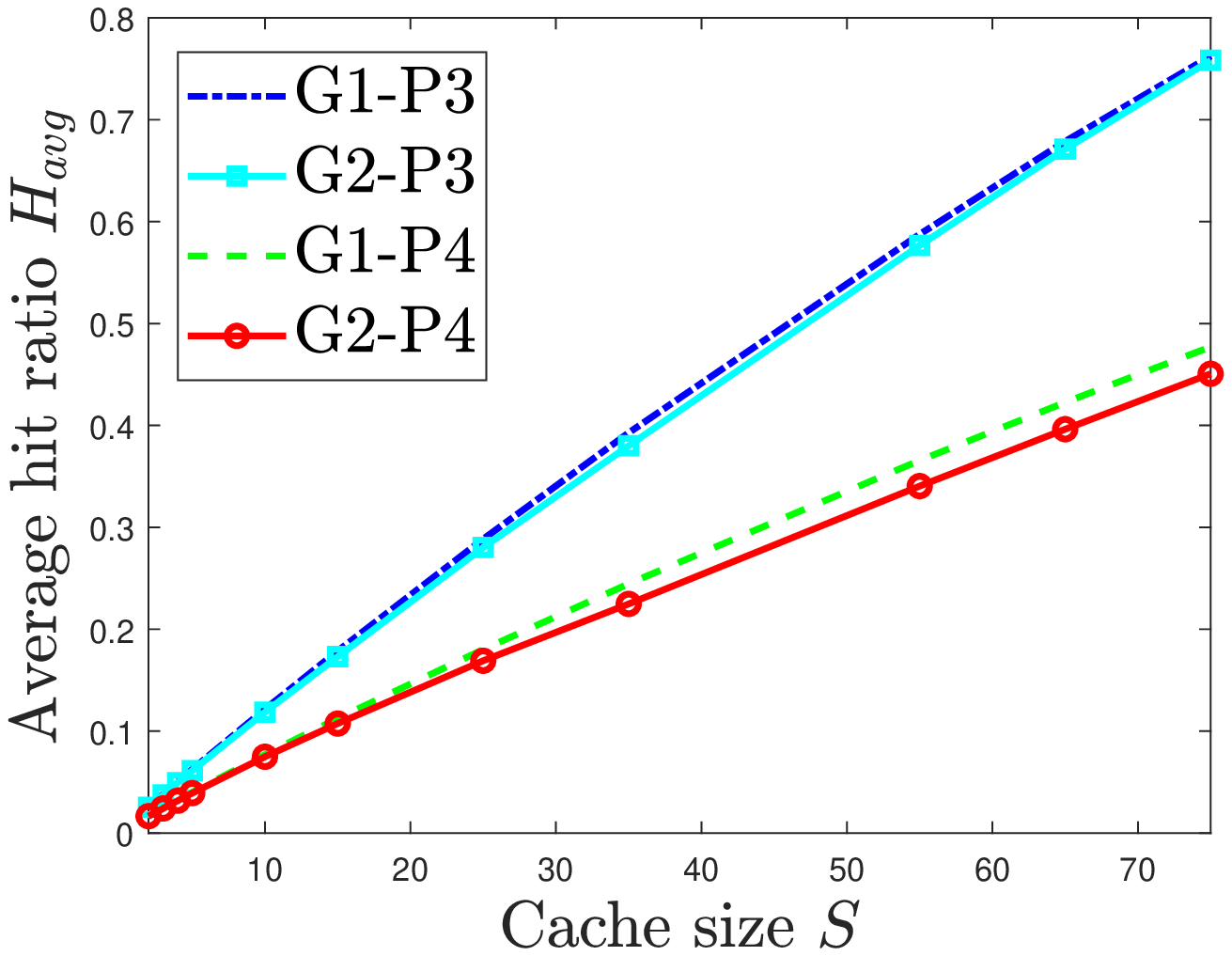}}
  \subfigure[Number of users, $N$]{\includegraphics[width=.49\linewidth]{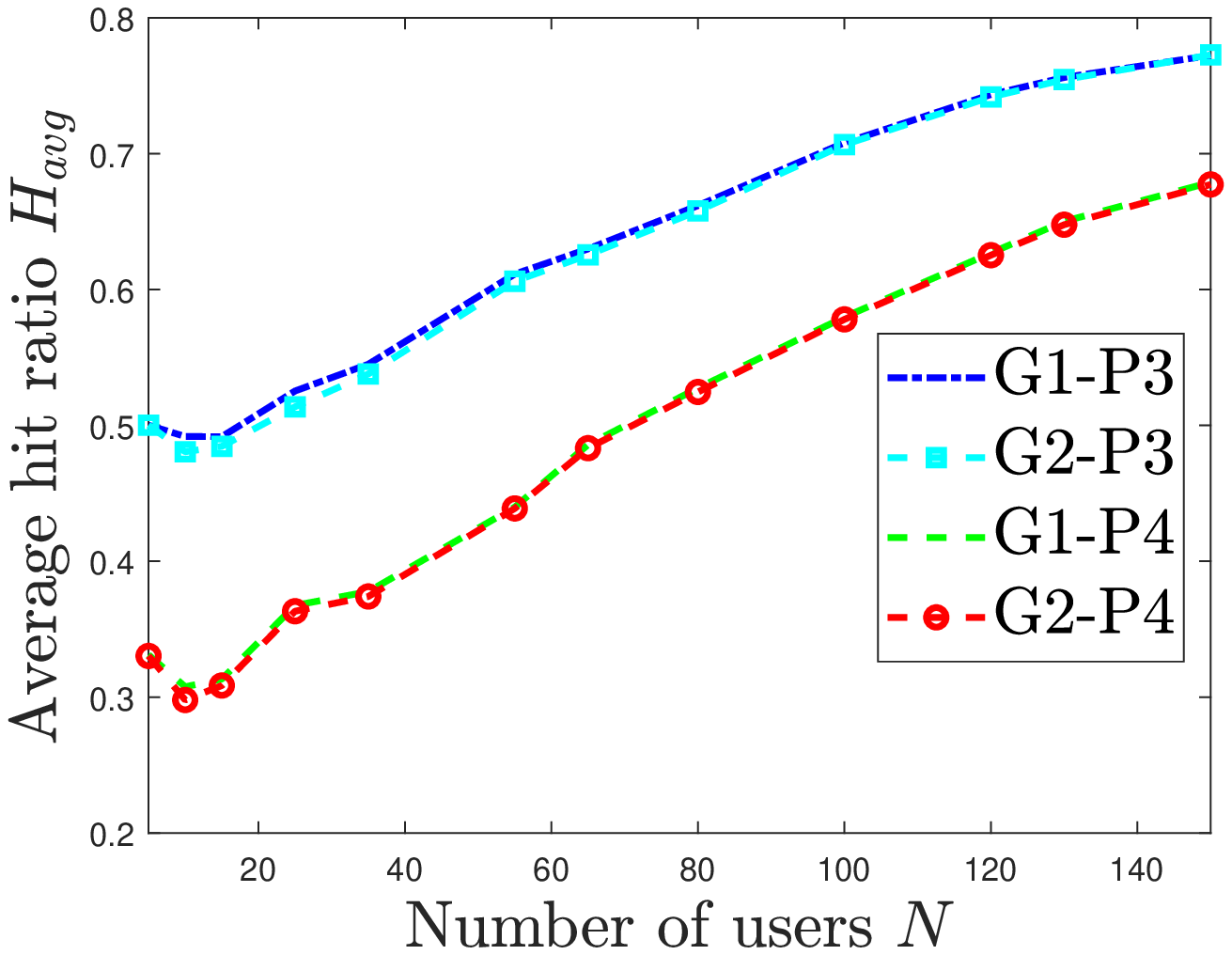}}
  \subfigure[Library size $F$]{\includegraphics[width=.49\linewidth]{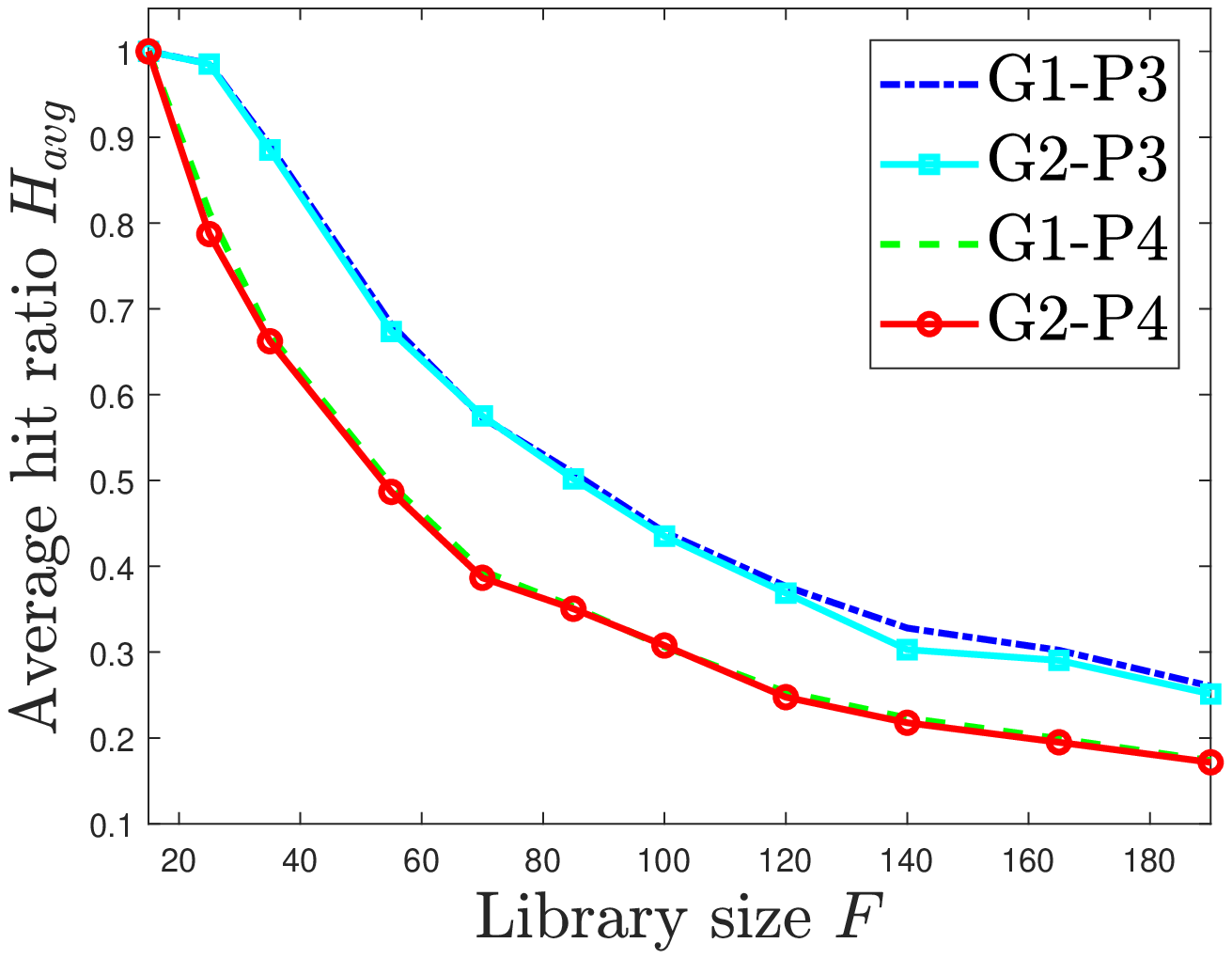}}
  \subfigure[Collaboration distance $\mathcal{D}$]{\includegraphics[width=.49\linewidth]{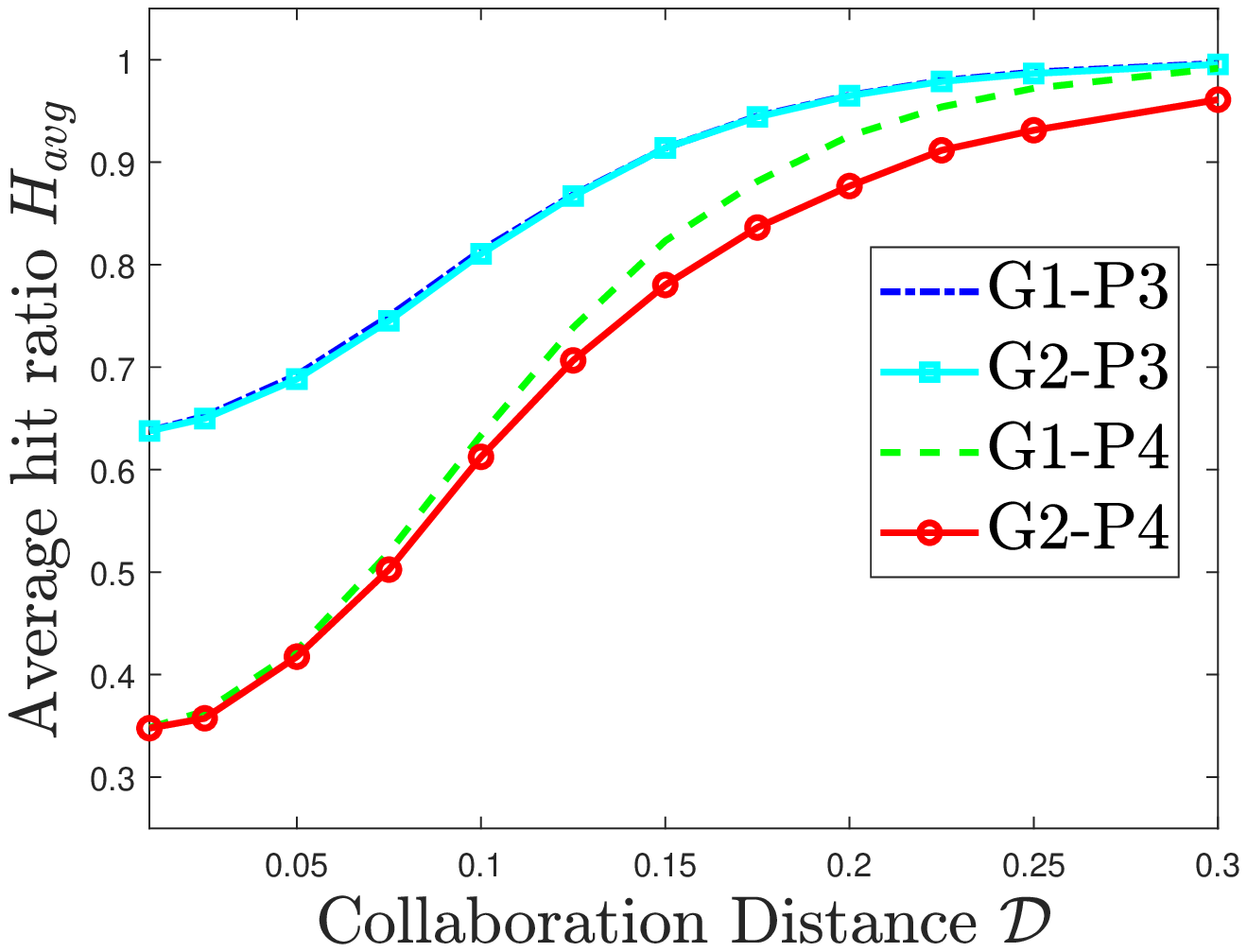}}
  \caption{The average hit ratio $H_{avg}$ according to system parameters, where algorithms {\bf G1-P3}, {\bf G2-P3}, {\bf G1-P4}, and {\bf G2-P4} are compared.}
\label{fig:wrt}
\end{figure}

In Fig.~\ref{fig:wrt}d, the average hit ratio $H_{avg}$ versus the collaboration distance $\mathcal{D}$ is illustrated, where $S=20$, $F=70$, and $N=50$. It is seen that the performance difference between the two caching strategies is noticeable especially for small values  of $\mathcal{D}$ as depicted in the figure. This is due to the fact that increasing $\mathcal{D}$ leads to an increment of the user density within $\mathcal{D}$, which can raise the chance such that a requesting user can interact with more users to retrieve their content---the contact probability $p_{u_i,u_j}$ that two users $u_i$ and $u_j$ are within the communication distance $\mathcal{D}$ increases with the user density within $\mathcal{D}$, thus resulting in the enhancement on the average hit ratio $H_{avg}$. Intuitively, when $\mathcal{D}$ is higher, since a user is capable of interacting with more users, the performance difference between the two caching strategies becomes diminished as depicted in Fig.~\ref{fig:wrt}d. Thus, increasing $\mathcal{D}$ eventually leads to the case where the performance of both caching strategies approaches the upper bound $H_{avg}=1.0$.

From Fig.~\ref{fig:wrt}, we conclude that the caching strategies employing the personalized file preferences performs significantly better than the baseline caching strategies. We also observe that the performance of the {\bf G2} algorithm is quite comparable to or slightly lower than that of the {\bf G1} algorithm, whereas the complexity of the {\bf G2} algorithm can further be reduced (see Remarks \ref{R:2} and \ref{R:4}).

\subsubsection{The Effects of Predicted Ratings}
Furthermore, it is worthwhile to examine the effects of predicted ratings on the hit ratio $H_{avg}$. To this end, we generate another rating matrix based on the pre-filtered $333\times261$ rating matrix (namely, the {\em non-masked} rating matrix) as follows. Another pre-filtered rating matrix (the {\em masked} rating matrix) is obtained by randomly masking $M\%$ of the total given ratings in the original pre-filtered rating matrix. For example, when $M=10\%$, the sparsity of the masked $333\times261$  rating matrix becomes 50.9\%. We then predict all the missing ratings in both non-masked and masked pre-filtered rating matrices via CF, where the case using the non-masked rating matrix provides an upper bound on the hit ratio performance. The rest of environmental settings is the same as those in Section~\ref{section:611}.

\begin{figure}[t]
\centering
 \includegraphics[width=0.6\linewidth]{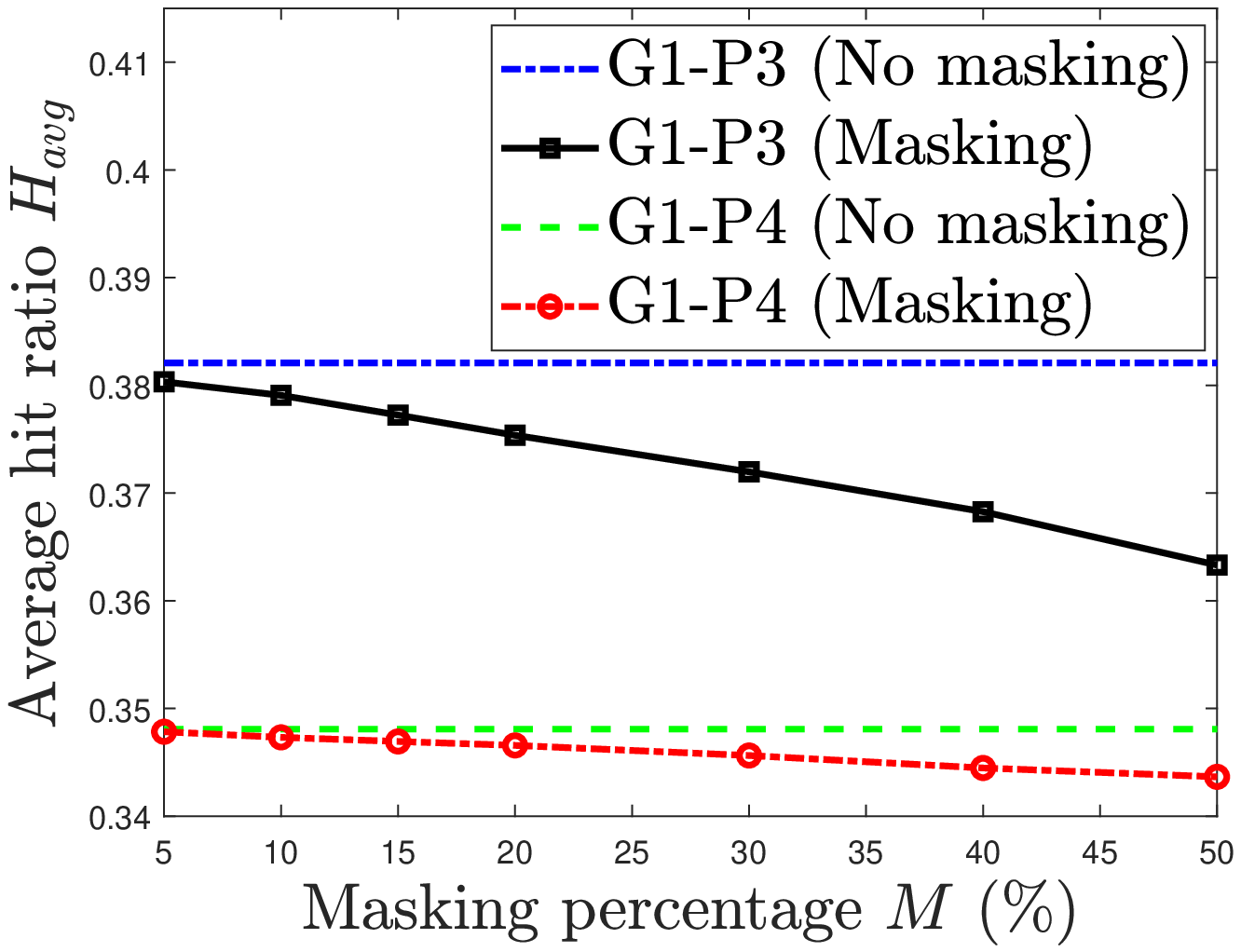}
\caption{The average hit ratio $H_{avg}$ versus the masking percentage $M$ (\%), where four cases {\bf G1-P3 (No masking)}, {\bf G1-P3 (Masking)}, {\bf G1-P4 (No masking)}, and {\bf G1-P4 (Masking)} are compared.}
\label{fig:wrtM}
\end{figure}
In Fig.~\ref{fig:wrtM}, the average hit ratio $H_{avg}$ versus the masking percentage $M$ (\%) is illustrated, where $S=20$, $F=80$, $N=60$, and $\mathcal{D}=0.1$. From this figure, the following observations are found: the performance of {\bf G1-P3 (Masking)} and {\bf G1-P4 (Masking)} (i.e., caching based on the masked rating matrix) is degraded with increasing $M$; {\bf G1-P3 (Masking)} is still quite superior to {\bf G1-P4 (No masking)} (i.e., its counterpart scheme based on the non-masked rating matrix); and the performance gap between {\bf G1-P4 (No masking)} and {\bf G1-P4 (Masking)} (the case with the common file preferences) is negligible. The last observation comes from the fact that the file popularity in~\eqref{eq:pop} is calculated by summing up the predicted ratings for a file $f_a\in \mathcal{F}$ over all $N$ users and thus prediction errors are averaged out.

As addressed before, we are capable of further improving the prediction accuracy by using not only more sophisticated CF but also data imputation techniques. Such approaches enable us to greatly enhance the hit ratio performance of the caching strategies based on the masked rating matrix (i.e., {\bf G1-P3 (Masking)} and {\bf G1-P4 (Masking)}).


 \begin{figure}[t]
\centering
 \includegraphics[width=0.6\linewidth]{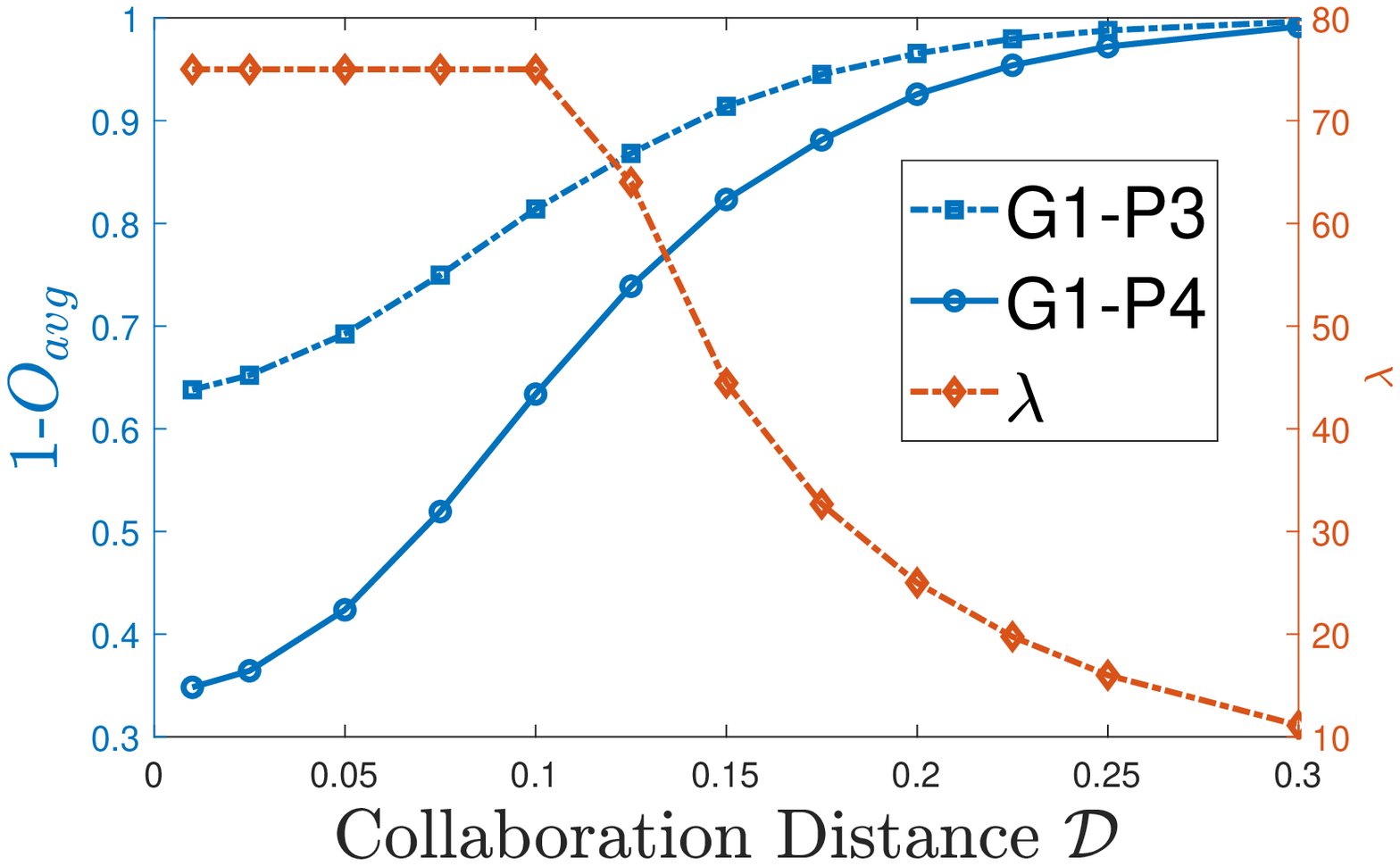}
\caption{The trade-off between $\lambda$ (right) and $1- O_{avg}$ (left) according to the collaboration distance $\mathcal{D}$, where algorithms {\bf G1-P3} and {\bf G1-P4} are compared.}
\label{fig:tradeoffwrtCD}
\end{figure}

\begin{figure}[t]
\centering
 \includegraphics[width=0.6\linewidth]{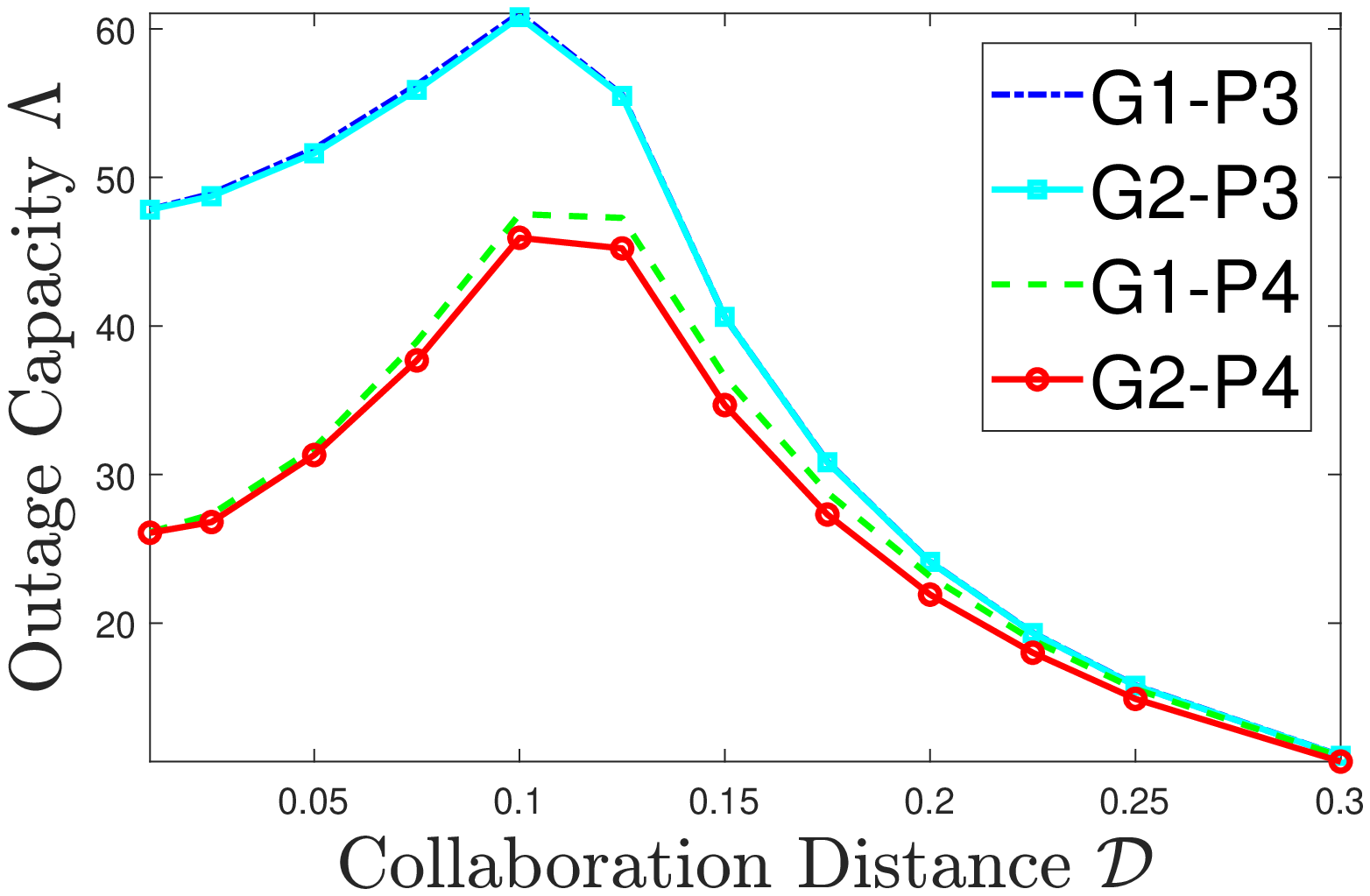}
\caption{The outage capacity $\Lambda$ versus the collaboration distance $\mathcal{D}$, where algorithms {\bf G1-P3}, {\bf G2-P3}, {\bf G1-P4}, and {\bf G2-P4} are compared.}
\label{fig:OCwrtCD}
\end{figure}

\subsubsection{Evaluation in terms of the Outage Capacity} Let us turn to evaluating the performance on the outage capacity $\Lambda$ by numerically showing that our caching strategy employing the personalized file preferences offers substantial gains over the case with the common file preferences.

In Fig.~\ref{fig:tradeoffwrtCD}, we illustrate the trade-off between $\lambda$ in \eqref{eq:oc} and $1- O_{avg}$ in \eqref{eq:oc} according to the collaboration distance $\mathcal{D}$ for $S=20$, $F=70$, and $N=150$, where $\lambda$ is the maximum number of simultaneous single-hop D2D communication links in each time slot and $O_{avg}$ is the average outage probability. From the figure, it is found that increasing $\mathcal{D}$ leads to an increment of the term $1-O_{avg}$ due to the increased user density within $\mathcal{D}$, but reduces the term $\lambda$ due to less D2D communications links over the network.
 
 Furthermore, it  is worth noting that the collaboration distance $\mathcal{D}$ should be decided carefully depending on the performance metric; if we aim at maximizing the average hit ratio $H_{avg}$, then $\mathcal{D}$ should be as large as possible as depicted in Fig.~\ref{fig:wrt}d, but this setting leads to a small $\lambda$, thus resulting in the reduction on the outage capacity $\Lambda$. To numerically find the optimal $\mathcal{D}^*$ in terms of $\Lambda$, in Fig.~\ref{fig:OCwrtCD}, we illustrate the outage capacity $\Lambda$ versus the collaboration distance $\mathcal{D}$, where $S=20$, $F=70$, and $N=150$. From the figure, the following insightful observations are made: 1) the proposed strategy ({\bf G1-P3} and {\bf G2-P3}) consistently outperforms its counterpart scheme ({\bf G1-P4} and {\bf G2-P4}); 2) the optimal $\mathcal{D}^*$ is found at $\mathcal{D}=0.1$ for two different caching strategies; and 3) the performance difference between two caching strategies is significant especially for small values of $\mathcal{D}$.

\subsubsection{Extension to a Mobile Hybrid Network}~\label{section:hb}  We extend our caching strategy to a mobile hybrid network consisting of $L$ static femto-cell base stations (FBSs), denoted by the set $\mathcal{L}=\{l_1,l_2,\cdots,l_L\}$, and $N$ mobile users. Each FBS $l_j\in\mathcal{L}$ is assumed to be equipped with a local cache of the same finite storage capacity $S_{FBS}$, where $S_{FBS}\geq S$ due to a physically larger storage size at each FBS. As in~\cite{wyshin}, the network area is divided into $L$ square cells of equal size, each of which has one FBS placed in the center of the belonging cell. During the content placement process, each FBS $l_i \in \mathcal{L}$ stores the content files according to a caching vector $\textbf{C}_{l_i}= [c_{l_i,f_1}, c_{l_i,f_2}, \cdots, c_{l_i,f_F}]$, where $c_{l_i,f_j}$ is $1$ if the FBS $l_i$ caches file $f_j$ and $0$ otherwise. In order to have a feasible cache allocation strategy $\mathcal{C}_{FBS}^{L\times F}=$ $ [\textbf{C}_{l_1}; \textbf{C}_{l_2}; \cdots; \textbf{C}_{l_L}]$ for FBSs, $\mathcal{C}_{FBS}^{L\times F}$ should satisfy $\sum_{f_j \in \mathcal{F}} c_{l_i,f_j} \leq S_{FBS}$ and $c_{l_i,f_j} \in \left\{0,1\right\}$, $\forall l_i \in \mathcal{L}$, $f_j \in \mathcal{F}$. In the delivery phase, each user retrieves the requested content file either from his/her own cache or via single-hop communications from one of the mobile users or static FBSs storing the requested content file in their caches within the collaboration distance $\mathcal{D}$ of the requested user. In the mobile hybrid network, we present two caching strategies as follows.

 \begin{itemize}
    \item {\bf A na\"{\i}ve strategy}: Each FBS $l_j \in \mathcal{N}$ stores the $S_{FBS}$ most popular content files based on the inferred file popularity $P(f_a)$. Then, each element $c_{l_j,f_i}$ of  a feasible cache allocation strategy $\mathcal{C}_{FBS}^{L\times F}$ is given by 
\begin{equation} \label{eq:c_l}
c_{l_j,f_i} = \begin{cases}
    1 &  f_i \in \mathcal{F}_1 \\
    0 & $otherwise$,
\end{cases}
\end{equation}
where $\mathcal{F}_1 \subset \mathcal{F}$ is the set of the $S_{FBS}$ most popular content files. The caching strategy for mobile users is designed by solving the Problem \textbf{[P3]} without the knowledge of the cached content at FBSs. 
\item {\bf A joint optimization strategy}: We aim at showing a joint cache allocation strategy $\mathcal{\tilde{C}}^{(L+N)\times F}$ $= [\mathcal{C}_{FBS}^{L\times F}; \mathcal{C}^{N\times F}]$ by solving the following optimization problem:
\begin{subequations} \label{eq:opph}
\begin{align} \label{eq:ofph}
\textbf{[P5]} &\underset{\mathcal{\tilde{C}}^{(L+N)\times F}}{\max}  \sum_{u_i \in \mathcal{N}} \sum_{f_a \in \mathcal{F}} P(u_i) P(f_a|u_i) \cdot \nonumber \\ & \Biggl(\!\! 1\!-\!\!\!\!\! \prod_{u_j \in \mathcal{N}}\!\!\!  \left(\!1-\!\!c_{u_j, f_a}p_{u_i,u_j}\!\right)\!\!\! \prod_{l_j \in \mathcal{L}}\!\!\left(\!1-\!\! c_{l_j, f_a}p_{u_i,l_j}\!\right)\!\!\Biggr)\!\!\!\!
\end{align}
subject to
\begin{align}
\sum_{f_a \in \mathcal{F}}c_{u_i,f_a} \leq S, \hspace{2.2cm} \forall u_i \in \mathcal{N},
\end{align}
\begin{align}
\sum_{f_a \in \mathcal{F}}c_{l_j,f_a} \leq S_{FBS}, \hspace{1.7cm} \forall l_j \in \mathcal{L},
\end{align}
\begin{align}
 c_{u_i,f_a}, c_{l_j,f_a} \! \in \!\left\{0,1\right\}, \  \forall u_i\! \in\! \mathcal{N}, l_j\! \in\! \mathcal{L}, f_a\! \in\! \mathcal{F}.
\end{align}
\end{subequations}
Note that the problem \textbf{[P5]} can be solved using almost the same algorithms presented in Section 5. 
\end{itemize}
Intuitively, the performance of the joint optimization strategy should be better than the na\"{\i}ve one, but this comes at the price of increase in the complexity. In the following, we evaluate the hit ratio for the following four algorithms: the na\"{\i}ve strategy with the {\bf G1} algorithm for solving (10) ({\bf N-G1-P3}), the na\"{\i}ve strategy with the {\bf G2} algorithm for solving (10) ({\bf N-G2-P3}), the joint optimization strategy with the {\bf G1} algorithm for solving~\eqref{eq:opph} ({\bf J-G1-P5}), and the joint optimization strategy with the {\bf G2} algorithm for solving~\eqref{eq:opph} ({\bf J-G2-P5}).

\begin{figure}[t]
  \centering
  \subfigure[Number of FBSs, $L$, for $S_{FBS}=30$]{\includegraphics[width=.49\linewidth]{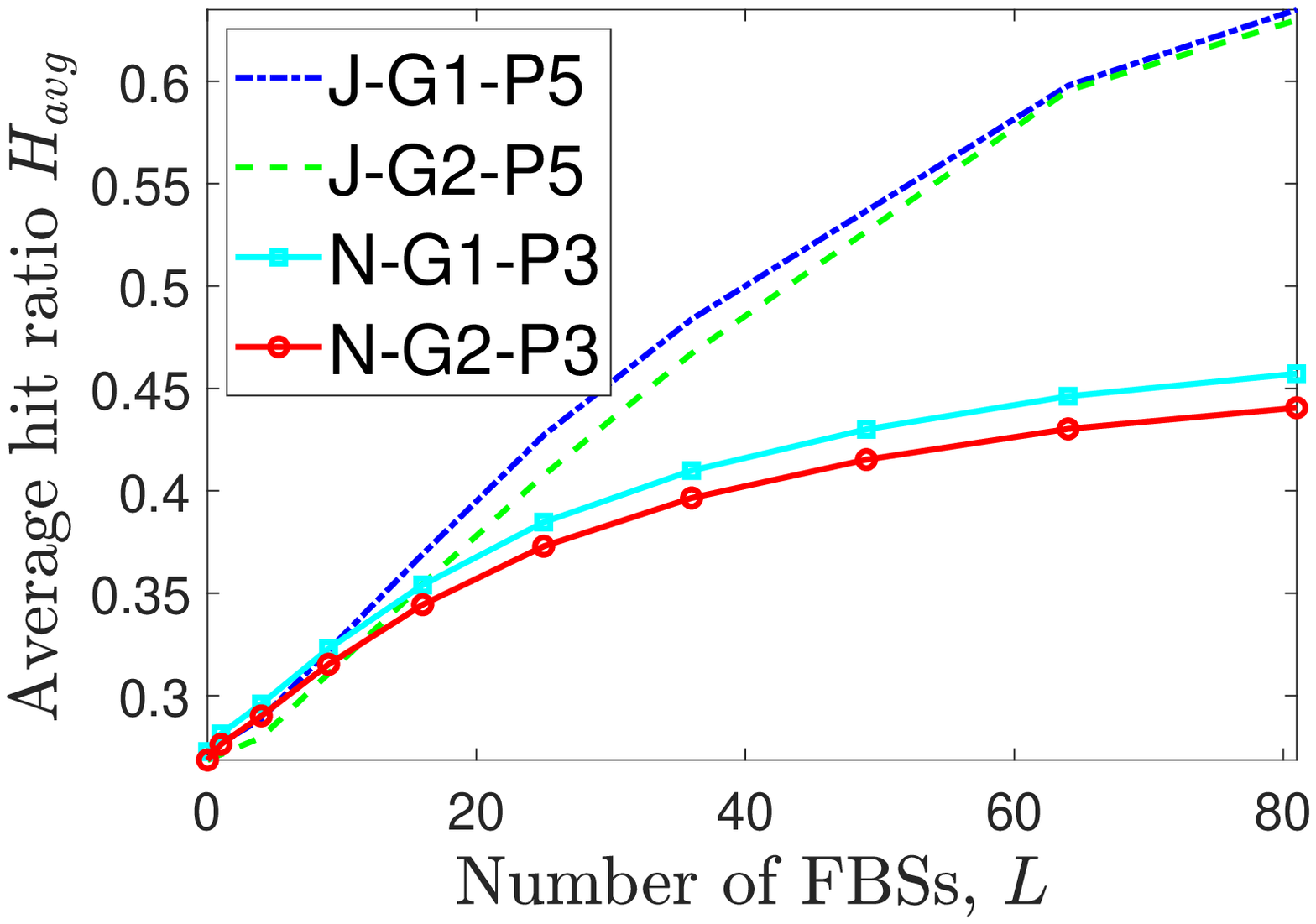}}
  \subfigure[Cache size $S_{FBS}$ for $L=25$]{\includegraphics[width=.49\linewidth]{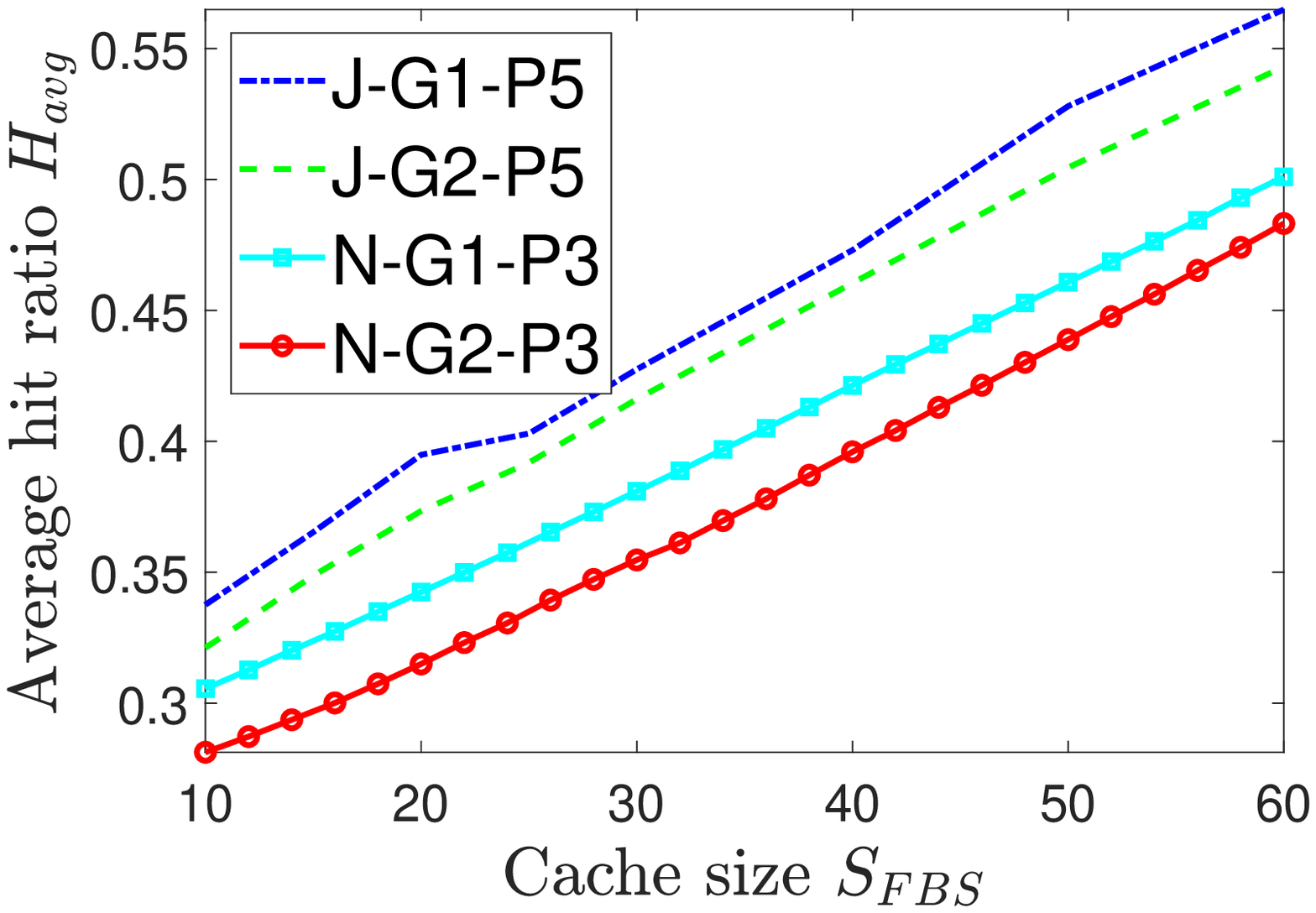}}
  \caption{The average hit ratio $H_{avg}$ for $N=60$, $F=100$, $S=10$, and $\mathcal{D}=0.1$, where algorithms {\bf N-G1-P3}, {\bf N-G2-P3}, {\bf J-G1-P5}, and {\bf J-G2-P5} are compared.}
\label{fig:wrtL}
\end{figure}

In Fig.~\ref{fig:wrtL}a, the average hit ratio $H_{avg}$ versus the number of FBSs, $L$, is illustrated, where $S_{FBS}=30$, $S=10$, $N=60$, $F=100$, and $\mathcal{D}= 0.1$. It is observed that the hybrid mobile network leads to substantial gains over the pure mobile network (i.e., $L=0$). The joint optimization strategy ({\bf J-G1-P5} and {\bf J-G2-P5}) consistently outperforms the na\"{\i}ve strategy ({\bf N-G1-P3} and {\bf N-G2-P3}), where the performance gap between the two strategies tends to get larger with increasing $L$ as depicted in the figure. In Fig.~\ref{fig:wrtL}b, the average hit ratio $H_{avg}$ versus the FBS's cache size $S_{FBS}$ is illustrated, where $L=25$, $N=60$, $F=100$, $S=10$, and $\mathcal{D}= 0.1$. As expected, $H_{avg}$ is enhanced with increasing $S_{FBS}$, and thus the deployment of cache-enabled FBSs brings a significant gain.

\begin{figure}[t]
  \centering
  \subfigure[Cache size $S$ for $N=30$, $F=200$, and $\mathcal{D}=0.05$]{\includegraphics[width=.49\linewidth]{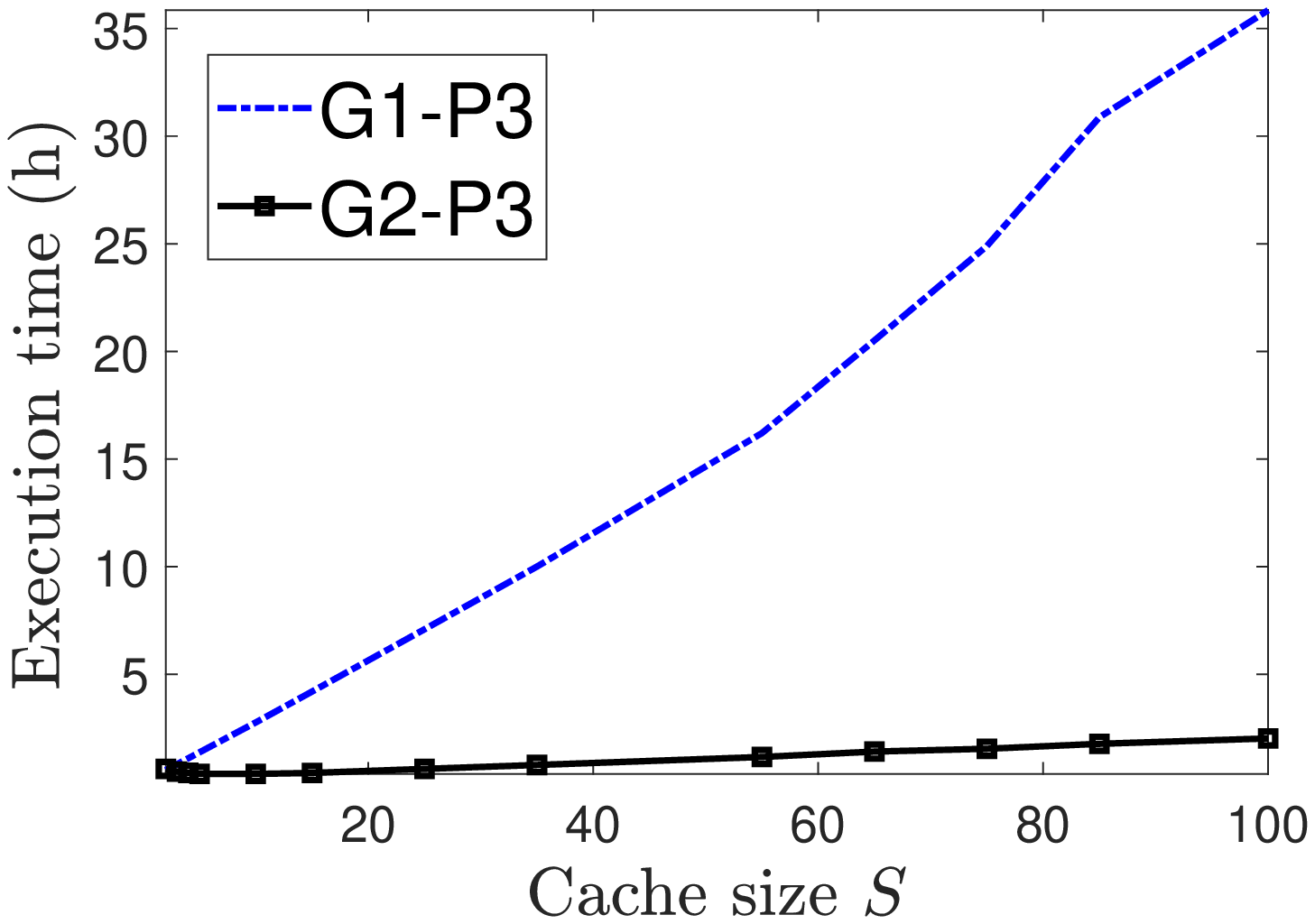}}
  \subfigure[Number of users, $N$ for $S=10$, $F=50$, and $\mathcal{D}=0.1$]{\includegraphics[width=.49\linewidth]{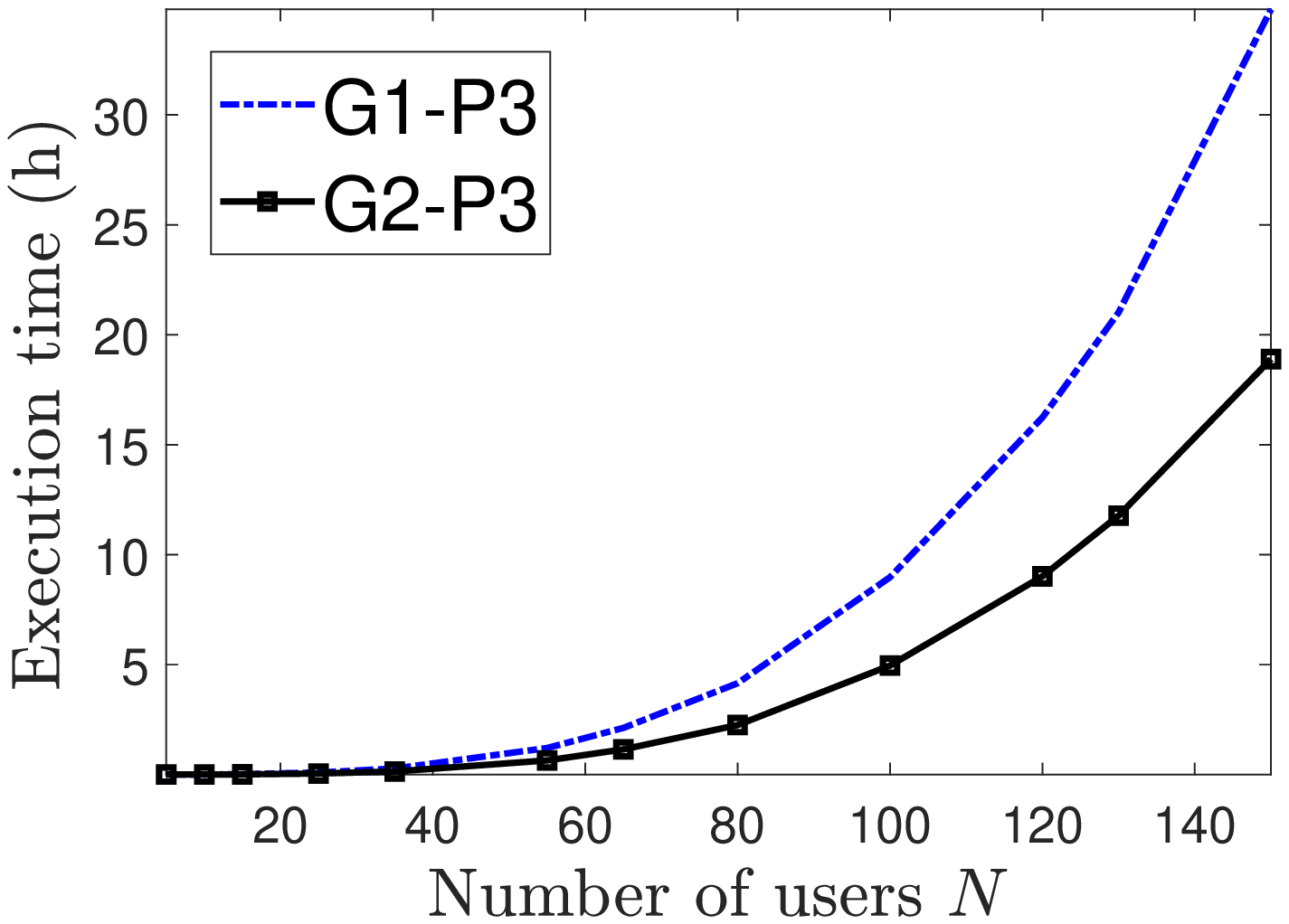}}
   \caption{The execution time according to system parameters, where algorithms {\bf G1-P3} and {\bf G2-P3} are compared.}
\label{fig:Twrt}
\end{figure}

\subsubsection{Empirical Evaluation of Complexity} We evaluate the execution time of the {\bf G1} and {\bf G2} Algorithms to validate our analytical claims in Remarks \ref{R:2} and \ref{R:4}, where the complexity of the {\bf G1} and {\bf G2} Algorithms are given by $\mathcal{O}(SN^3F^2)$ and $\mathcal{O}(N^3F^2)$, respectively. In Fig. \ref{fig:Twrt}a, the execution time of {\bf G1-P3} and {\bf G2-P3} versus the cache size $S$ is illustrated, where $N=30$, $F=200$, and $\mathcal{D}=0.05$. Our findings confirm that the runtime complexity of {\bf G1-P3} scales {\em linearly} with $S$ while the runtime complexity of {\bf G2-P3} hardly scales with $S$, which coincides with our analytical result. It is also seen that the complexity of the {\bf G2} Algorithm is dramatically reduced compared to the {\bf G1} Algorithm. In Fig. \ref{fig:Twrt}b, the execution time of both algorithms versus the number of users, $N$, is illustrated, where $S=10$, $F=50$, and $\mathcal{D}=0.1$. It is observed that the complexities of {\bf G1-P3} and {\bf G2-P3} are increasing with $N$, where they indeed asymptotically follow the {\em cubic} complexity in $N$, which is also consistent with our complexity analysis in Remarks \ref{R:2} and \ref{R:4}.

\section{Concluding Remarks}

This paper investigated the impact and benefits of personalized file preferences in a content-centric mobile network by proposing a CF-aided learning framework that enables us to infer model parameters based on the user rating history. The hit ratio under our mobile network was characterized by adopting single-hop-based D2D content delivery according to the two caching strategies employing both personalized file preferences and common file preferences. The hit ratio maximization problems were then reformulated into a submodular function maximization, and two scalable algorithms including a greedy approach were presented to significantly reduce the computational complexity. The proposed greedy algorithm was shown to achieve an approximate solution that is within a provable gap compared to the optimal solutions. In addition, data-intensive evaluation was performed using the MovieLens dataset, and it was demonstrated that the caching strategy employing the personalized file preferences has a significant performance gain over the case employing the common file preferences. More precisely, it was shown that the performance difference between the two caching strategies is significant under practical scenarios such that 1) the collaboration distance $\mathcal{D}$ is small, 2) the library size $F$ is large, 3) the storage capacity $S$ is neither too high to cache almost entire library nor too low not to be even capable of caching few content files, and 4) the density of users within $\mathcal{D}$ is not too high for users within $\mathcal{D}$ to collectively store the entire content in the library. Moreover, the superiority of our caching strategy employing the personalized file preferences was empirically validated in terms of the outage capacity. In addition, our caching strategies were extended to the hybrid mobile network to see the effects of deploying multiple FBSs on the hit ratio. The computational complexity was also demonstrated in comparison with our analytical results.


\ifCLASSOPTIONcompsoc
  \section*{Acknowledgments}
\else
  \section*{Acknowledgment}
\fi This work was supported by the European Research Council under the EU Horizon 2020 research and innovation program / ERC grant agreement no. 725929 (project DUALITY) and the National Research Foundation of Korea (NRF) grant funded by the Korea government (MSIT) (No.2019R1A2C2007982).


\ifCLASSOPTIONcaptionsoff
  \newpage
\fi

\end{document}